%% file: main.tex
\newcommand{\CLASSINPUTbaselinestretch}{1.0}
\documentclass[draftcls, onecolumn, 11pt]{IEEEtran}

\usepackage{cite}
\usepackage{amsmath,amssymb,amsfonts}
\usepackage{algorithmic}
\usepackage{graphicx}
\usepackage{textcomp}
\usepackage[dvipsnames,table]{xcolor}
\usepackage{amsthm}
\usepackage{tikz}
\usetikzlibrary{colorbrewer}
\usepackage{scalerel}
\usepackage{arydshln}
\usepackage{hyperref} 
\hypersetup{
  linkcolor  = Violet!85!black,
  citecolor  = YellowOrange!85!black,
  urlcolor   = Aquamarine!85!black,
  colorlinks = true,
}
\usepackage{ifthen}
\usepackage{xfrac}
\usepackage{xspace}
\usepackage{subcaption}
\usepackage[capitalize,noabbrev,nameinlink]{cleveref}
\crefname{linearprogram}{linear program}{linear programs}
\crefname{ineq}{inequality}{inequalities}

\newtheorem{proposition}{Proposition}
\input{macros}
\begin{document}
\title{Bandwidth Cost of Code Conversions in Distributed Storage: Fundamental Limits and Optimal Constructions}

\author{
\IEEEauthorblockN{Francisco Maturana and K. V. Rashmi}

\IEEEauthorblockA{Computer Science Department\\
Carnegie Mellon University\\
Pittsburgh, PA, USA \\
Email: fmaturan@cs.cmu.edu, rvinayak@cs.cmu.edu}
}

\maketitle

\begin{abstract}
  \input{abstract}
\end{abstract}
\input{sections/1-introduction}
\input{sections/2-background}
\input{sections/3-related_work}
\input{sections/4-model}
\input{sections/5-merge_regime_lower_bound}
\input{sections/6-merge_regime_construction}
\input{sections/7-savings}

\input{sections/8-conclusion}

\bibliographystyle{ieeetr}
\bibliography{main}

\end{document}

%% file: macros.tex
\colorlet{color1}{Dark2-D}
\colorlet{color2}{Dark2-E}
\colorlet{color3}{Dark2-C}
\colorlet{color4}{Dark2-B}
\newtheorem{theorem}{Theorem}
\newtheorem{lemma}[theorem]{Lemma}

\theoremstyle{definition}
\newtheorem{definition}{Definition}
\newtheorem{example}{Example}
\theoremstyle{remark}
\newtheorem{remark}{Remark}
\newcommand{\etal}{et al.\xspace}

\newcommand{\eg}{e.g.,\xspace}

\newcommand{\codename}{convertible code\xspace}
\newcommand{\codenames}{\codename{}s\xspace}
\newcommand{\mds}{MDS\xspace}

\newcommand{\combination}{merge\xspace}

\newcommand{\combinationregime}{\combination\ regime\xspace}

\newcommand{\accesscost}{access cost\xspace}
\newcommand{\bw}{bandwidth\xspace}
\newcommand{\bwcost}{bandwidth cost\xspace}
\newcommand{\bandwidthcost}{bandwidth cost\xspace}
\newcommand{\convbw}{conversion bandwidth\xspace}
\newcommand{\accessoptimal}{access-optimal\xspace}
\newcommand{\bandwidthoptimal}{bandwidth-optimal\xspace}
\newcommand{\bwoptimal}{\bandwidthoptimal}

\newcommand{\stripe}{codeword\xspace}
\newcommand{\stripes}{\stripe{}s\xspace}

\newcommand{\regimeone}{Regime 1\xspace}
\newcommand{\regimetwo}{Regime 2\xspace}
\newcommand{\infograph}{information flow graph\xspace}

\newcommand{\subsyms}{subsymbols\xspace}
\newcommand{\piggyfw}{Piggybacking framework\xspace}

\newcommand{\V}[1]{\mathbf{#1}}
\newcommand{\Field}[1]{\mathbb{F}_{#1}}
\newcommand{\Code}{\mathcal{C}}
\newcommand{\convertible}[2]{$({#1},{#2})$-convertible\xspace}
\newcommand{\convertibleDefault}{\convertible{\Fn}{\Fk}}
\newcommand{\accessConvertible}[2]{$({#1},{#2})$-access-optimally convertible\xspace}
\newcommand{\accessConvertibleDefault}{\accessConvertible{\Fn}{\Fk}}
\newcommand{\bwConvertible}[2]{$({#1},{#2})$-bandwidth-optimally convertible\xspace}
\newcommand{\bwConvertibleDefault}{\bwConvertible{\Fn}{\Fk}}

\newcommand{\Initial}{I}
\newcommand{\Final}{F}
\newcommand{\ICode}{{\Code^{\Initial}}}
\newcommand{\FCode}{{\Code^{\Final}}}

\newcommand{\Indices}{\mathcal{I}}

\newcommand{\ParamCombCode}[4]{\(({#1},\allowbreak{#2};\allowbreak{#3},\allowbreak{#4})\) \codename}
\newcommand{\ParamCombCodeDefault}{\ParamCombCode{\In}{\Ik}{\Fn}{\Cs\Ik}}
\newcommand{\ParamCode}[4]{\(({#1},\allowbreak{#2};\allowbreak{#3},\allowbreak{#4})\) \codename}
\newcommand{\MdsParamCode}[4]{\(({#1},\allowbreak{#2};\allowbreak{#3},\allowbreak{#4})\) MDS \codename}

\newcommand{\ParamCodeDefault}{\ParamCode{\In}{\Ik}{\Fn}{\Fk}}
\newcommand{\MdsParamCodeDefault}{\MdsParamCode{\In}{\Ik}{\Fn}{\Fk}}

\newcommand{\Partition}{\mathcal{P}}
\newcommand{\IPart}{\Partition_{\Initial}}
\newcommand{\FPart}{\Partition_{\Final}}

\newcommand{\Set}[2]{{P^{#1}_{#2}}}
\newcommand{\Iset}[1]{{\Set{\Initial}{#1}}}
\newcommand{\Fset}[1]{{\Set{\Final}{#1}}}

\newcommand{\In}{n^{\Initial}}
\newcommand{\Ik}{k^{\Initial}}
\newcommand{\Ir}{r^{\Initial}}

\newcommand{\Is}{\lambda^{\Initial}}

\newcommand{\Fn}{n^{\Final}}
\newcommand{\Fk}{k^{\Final}}
\newcommand{\Fr}{{r^{\Final}}}

\newcommand{\Fs}{\lambda^{\Final}}

\newcommand{\Irt}{\widetilde{r}^{\Initial}}
\newcommand{\Frt}{\widetilde{r}^{\Final}}
\newcommand{\NNodes}{M}

\newcommand{\Cs}{\varsigma}

\newcommand{\size}{M}

\newcommand{\blocksize}{B}

\newcommand{\Msg}{\V{m}}

\newcommand{\Ibase}{\ICode'}
\newcommand{\puncIbase}{\ICode''}
\newcommand{\Fbase}{\FCode'}
\newcommand{\Ivcode}{\ICode}
\newcommand{\Fvcode}{\FCode}

\newcommand{\IPP}[3]{\V{#2}^{\scaleto{(#3)}{6pt}} \V{p}_{#1}^I}
\newcommand{\IPPA}[3]{\textcolor{color4}{{} +  \V{#2}^{\scaleto{(#3)}{6pt}} \V{p}_{#1}^I}}
\newcommand{\IPPB}[3]{\textcolor{color5}{{} +  \V{#2}^{\scaleto{(#3)}{6pt}} \V{p}_{#1}^I }}

\newcommand{\FPP}[2]{\V{#2}\V{p}_{#1}^F}
\newcommand{\FPPA}[2]{\textcolor{color4}{\V{#2}\V{p}_{#1}^F}}
\newcommand{\FPPB}[2]{\textcolor{color5}{\V{#2}\V{p}_{#1}^F }}

\newcommand{\unchangedv}[3]{u_{{#1},{#2},{#3}}}
\newcommand{\retiredv}[2]{v_{{#1},{#2}}}

\newcommand{\newvs}[2]{w_{{#1},{#2}}}
\newcommand{\sourcev}[1]{s_{#1}}

\newcommand{\sinkvs}[1]{t_{#1}}
\newcommand{\centralv}{c}

\newcommand{\unchangedvset}[2]{\mathcal{U}_{{#1},{#2}}}
\newcommand{\retiredvset}[1]{\mathcal{R}_{#1}}

\newcommand{\newvsets}[1]{\mathcal{N}_{#1}}

\newcommand{\dl}[1]{\beta\!\left({#1}\right)}
\newcommand{\Ssize}[1]{\tilde{r}_{#1}}

\newcommand{\Icc}[2]{c^{\Initial}_{{#1},{#2}}}
\newcommand{\Fcc}[2]{c^{\Final}_{{#1},{#2}}}
\newcommand{\Mcc}[2]{\Msg^{({#1})}_{#2}}

\DeclareMathOperator{\Lcm}{lcm}

\makeatletter
\DeclareRobustCommand{\cvdots}{%
\vcenter{%
  \baselineskip=4pt%
  \hbox{.}\hbox{.}\hbox{.}
}}
\makeatother

%% file: abstract.tex
Erasure codes have become an integral part of distributed storage systems as a tool for providing data reliability and durability under the constant threat of device failures.
In such systems, an $[n, k]$ code over a finite field $\Field{q}$ encodes $k$ message symbols from $\Field{q}$ into $n$ codeword symbols from $\Field{q}$ which are then stored on $n$ different nodes in the system.
Recent work has shown that significant savings in storage space can be obtained by tuning $n$ and $k$ to variations in device failure rates.
Such a tuning necessitates \emph{code conversion}: the process of converting already encoded data under an initial $[\In, \Ik]$ code to its equivalent under a final $[\Fn, \Fk]$ code.
The default approach to conversion is to re-encode the data under the new code, which places significant burden on system resources.
\emph{Convertible codes} are a recently proposed class of codes for enabling resource-efficient conversions.
Existing work  on convertible codes has focused on minimizing the access cost, i.e., the number of code symbols accessed during conversion. 
Bandwidth, which corresponds to the amount of data read and transferred, is another important resource to optimize during conversions. 

In this paper, we initiate the study on the fundamental limits on bandwidth used during code conversion and present constructions for \bandwidthoptimal \codenames.
First, we model the code conversion problem using network information flow graphs with variable capacity edges. 
Second, focusing on MDS codes and an important parameter regime called the merge regime, we derive tight lower bounds on the bandwidth cost of conversion. 
The derived bounds show that the bandwidth cost of conversion can be significantly reduced even in regimes where it has been shown that \accesscost cannot be reduced as compared to the default approach.
Third, we present a new construction for MDS convertible codes which matches the proposed lower bound and is thus \bandwidthoptimal during conversion.

%% file: sections/1-introduction.tex
\section{Introduction}
Erasure codes are an essential tool in distributed storage systems used to add redundancy to data in order to avoid data loss when device failures occur~\cite{ghemawat2003google,borthakurhdfs,huang2012erasure,asfapache}.
In particular, Maximum Distance Separable (MDS) codes are widely used for this purpose in practice because they require the minimum amount of storage overhead for a given level of failure tolerance.
In this setting, an $[n, k]$ MDS code over a finite field $\Field{q}$ is used to encode a message consisting of $k$ symbols of $\Field{q}$ into a codeword consisting of $n$ symbols of $\Field{q}$.\footnote{In the literature, this set of $n$ symbols is sometimes called a \emph{stripe} instead of a codeword. In this work, we make no distinctions between these two terms.}
Each of these $n$ codeword symbols are then stored on $n$ distinct nodes of the distributed storage system (typically, nodes correspond to storage devices residing on different servers).
Large-scale distributed storage systems usually comprise hundreds to thousands of nodes, while $n$ is much smaller in comparison, meaning that these systems store many such \stripes distributed across different subsets of nodes.
The MDS property ensures that any subset of $k$ symbols out of the $n$ symbols in the \stripe is enough to decode the original data. This provides tolerance for up to $(n - k)$ node failures.

The parameters $n$ and $k$ are typically set based on the reliability of storage devices and additional requirements on system performance and storage overhead.
Recent work by Kadekodi et al.~\cite{kadekodi2019cluster} has shown that the failure rate of disks can vary drastically over time, and that significant savings in storage space (and hence operating costs) can be achieved by tuning the code rate to the observed failure rates.
Such tuning typically needs to change both $n$ and $k$ of the code, due to other practical system constraints on these parameters~\cite{kadekodi2019cluster}.
Other reasons for tuning parameters include changing $k$ in response to changes in data popularity, and adapting the code rate to limit the total amount of storage space used.
Such tuning of parameters requires converting the already encoded data from one set of parameters to the newly chosen set of parameters.
The \emph{default approach} to achieving this is to re-encode, that is, read the encoded data, decode if necessary, re-encode it under the new code, and then write it back into the relevant nodes.
However, such an approach necessitates significantly high overhead in terms of network bandwidth, I/O, and CPU resources in the cluster.
This disrupts the normal operation of the storage system.

These applications have led to the study of the \emph{code conversion} problem~\cite{maturana2020convertible,maturana2020access}.
Code conversion is the process of transforming a collection of \stripes encoding data under an initial code $\ICode$ into a collection of \stripes encoding the same data under a final code $\FCode$.\footnote{The superscripts $\Initial$ and $\Final$ stand for initial and final, respectively.}
Given certain parameters for $\ICode$ and $\FCode$, the goal is to design the codes $\ICode$ and $\FCode$ along with a conversion procedure from $\ICode$ to $\FCode$ that is efficient in conversion (according to some notion of conversion cost as will be discussed subsequently).
The design is subject to additional decodability constraints on the codes $\ICode$ and $\FCode$, such as both satisfying the MDS property, since both these codes encode data in the storage system at different snapshots in time.
A pair of codes designed to efficiently convert encoded data from an $[\In, \Ik]$ code to an $[\Fn, \Fk]$ code is called an \emph{\ParamCodeDefault}, and the initial $[\In, \Ik]$ code is said to be \emph{\convertibleDefault}.
In practice, the exact value of the final parameters $\Fn$ and $\Fk$ might not be known at the time of code construction, as it might depend on future failure rates.
Instead, one might have some finite set of possible values for the pair $(\Fn, \Fk)$ that will be chosen from at the time of conversion.
For this reason, we will also seek to construct initial codes which are simultaneously \convertibleDefault for all $(\Fn, \Fk)$ in a given finite set of final parameter values.
This allows the flexibility to choose the parameters $\Fn$ and $\Fk$ at the time conversion is performed.

Existing works on convertible codes have studied efficiency in terms of the \emph{access cost} of conversion, which corresponds to the number of \stripe symbols accessed during conversion.
In particular, previous works~\cite{maturana2020convertible,maturana2020access} have derived tight lower bounds on the access cost of conversion for linear MDS convertible codes, and presented explicit constructions of MDS convertible codes that meet those lower bounds (i.e.\ access-optimal MDS convertible codes).
Another important resource overhead incurred during conversion is that on the network bandwidth, which we call \textit{conversion bandwidth}.
In the system, this corresponds to the total amount of data read and transferred between nodes during conversion.
Access-optimal convertible codes, by virtue of reducing the number of code symbols accessed, also reduce the network bandwidth cost as compared to the default approach.
However, \emph{it is not known if these codes are also bandwidth optimal}.

In this paper, the focus is on MDS convertible codes that incur minimum conversion bandwidth (i.e.\ bandwidth-optimal convertible codes).
We specifically focus on a parameter regime known as the \emph{merge regime}, which has been shown to play the most critical role in the analysis and construction of \codenames~\cite{maturana2020convertible}.
The merge regime corresponds to conversions where multiple initial \stripes are merged into a single final \stripe (i.e.\ $\Fk = \Cs\Ik$ for some integer $\Cs \geq 2$).

For the access cost of conversion in the merge regime, it is known~\cite{maturana2020access} that one cannot do better than the default approach for a wide range of parameters (specifically, when $(\In - \Ik) < (\Fn - \Fk)$, which we term \emph{\regimeone}).
For the remaining set of parameters  (which we term \emph{\regimetwo}), \accessoptimal \codenames lead to considerable reduction in access cost compared to the default approach. Yet, it is possible that there is room for a significant reduction in bandwidth cost in both of these regimes. 
This is possible by considering codes over finite extensions of finite fields $\Field{q^\alpha}$, where each \stripe symbol can be interpreted as an $\alpha$-length vector over the base field $\Field{q}$. {Such codes are called \textit{vector codes}.}
Vector codes allow conversion procedures to download elements of the base field from nodes, allowing them to download only a fraction of the \stripe symbols.
This is inspired by the work on regenerating codes by Dimakis \etal~\cite{dimakis2010network} who used vector codes to reduce bandwidth cost of reconstructing a subset of the \stripe symbols.

\noindent\textbf{Contributions of this paper.}
First, to analyze the \bwcost, we model the code conversion problem via a \textit{network information flow graph}.
This is a directed acyclic graph with capacities, where vertices represent nodes and edges represent the communication between nodes.
The approach of \infograph{}s has been used by Dimakis \etal~\cite{dimakis2010network} in the study on regenerating codes.
Unlike in the case of regenerating codes, the proposed model involves \textit{variable capacities on edges} representing data download during conversion.
This feature turns out to be be critical; we show that conversion procedures which download a uniform amount of data from each node are necessarily sub-optimal.

Second, by using the information flow model, we derive a tight lower bound on the network bandwidth cost of conversion for MDS convertible codes in the merge regime.
Specifically, we use the information flow graph to derive constraints on edge capacities that we then feed into an optimization problem whose objective is to minimize the bandwidth of conversion.
With this we derive a tight lower bound on the total bandwidth cost of conversion for given code parameters $(\In, \Ik; \Fn, \Fk)$.

Third, using the above derived (tight) lower bound, we show that (1) in \regimeone, where no reduction in \accesscost as compared to the default approach is possible, a substantial reduction in \bandwidthcost can be achieved, and (2) in \regimetwo, the \accessoptimal \codenames are indeed \bandwidthoptimal.

Fourth, we present an explicit construction of MDS convertible codes in the merge regime which achieves this lower bound and is therefore optimal in terms of  bandwidth cost.
This construction exploits the \emph{\piggyfw}~\cite{rashmi2017piggybacking}, which is a general framework for constructing vector codes, {and uses} access-optimal MDS convertible codes~\cite{maturana2020access} as a building block.

Above, only a single value of final parameters $\Fn$ and $\Fk$ was considered.
So finally, we propose a technique to transform our construction so as to be \textit{simultaneously} bandwidth-optimal in conversion for any given set of potential final parameter values. 
The proposed transformation exploits piggybacking in a recursive fashion.

\textbf{Organization.}
We review the necessary background and discuss related work in \cref{sec:background}.
In \cref{sec:model}, we describe our model of the conversion process as an information flow graph.
In \cref{sec:merge-bw}, we derive a lower bound on the conversion bandwidth of MDS convertible codes in the merge regime.
In \cref{sec:merge-bw:construction}, we propose an explicit construction for bandwidth-optimal MDS convertible codes in the merge regime, including the transformation to make the construction simultaneously bandwidth-optimal in conversion for multiple final parameter values.
In \cref{sec:evaluation}, we analyze the savings enabled by \bwoptimal \codenames.
We conclude the paper in \cref{sec:conclusion}.

%% file: sections/2-background.tex
\section{Background and related work}\label{sec:background}

In this section we start by introducing concepts from the existing literature that are used in this paper.
We then do an overview of other related work.

\subsection{Vector codes and puncturing}

An $[n, k, \alpha]$ vector code $\Code$ over a finite field $\Field{q}$ is an $\Field{q}$-linear subspace $\Code \subseteq \Field{q}^{\alpha n}$ of dimension $\alpha k$.
For a given \stripe $\V{c} \in \Code$ and $i \in [n]$, define $\V{c}_i = (c_{\alpha(i - 1) + 1}, \ldots, c_{\alpha i})$ as the $i$-th \emph{symbol} of $\V{c}$, which is a vector of length $\alpha$ over $\Field{q}$.
In the context of vector codes, we will refer to elements from the base field $\Field{q}$ as \emph{\subsyms}.
An \emph{encoding function} for $\Code$ is a function {$f(\Msg) = (f_1(\Msg), \ldots, f_n(\Msg))$} mapping messages $\Msg$ to \stripes of $\Code$.
We denote the encoding of a message $\Msg$ under a code $\Code$ as $\Code(\Msg)$.
An encoding function (or its associated code) is said to be \emph{systematic} if it always maps $\Msg$ to a \stripe having $\Msg$ as a prefix.
For an $[n, k, \alpha]$ vector code $\Code$, the encoding of message $\Msg \in \Field{q}^{k\alpha}$ is given by the mapping $\Msg \mapsto \Msg \V{G}$ where $\V{G} \in \Field{q}^{k\alpha \times n\alpha}$ is called the \emph{generator matrix} of $\Code$, and the columns of $\V{G}$ are called \emph{encoding vectors}.
An $[n, k, \alpha]$ vector code $\Code$ is maximum distance separable (MDS) if its minimum distance is the maximum possible:
\[
    \operatorname{min-dist}(\Code) = \min_{\V{c} \neq \V{c}' \in \Code} |\{ i \in [n] : \V{c}_i \neq \V{c}'_i \}| = n - k + 1.
\]
Equivalently, an $[n, k, \alpha]$ vector code $\Code$ is MDS if and only if for every $\V{c} \in \Code$, any $k$ symbols of $\V{c}$ uniquely specify the remaining $(n - k)$ symbols (i.e.\ every \stripe can be decoded from any $k$ symbols).
A \emph{scalar code} is a vector code with $\alpha = 1$.
We will omit the parameter $\alpha$ when it is clear from context or when $\alpha = 1$.
A \emph{puncturing} of a vector code $\Code$ is the resulting vector code after removing a fixed subset of symbols from every \stripe $\V{c} \in \Code$.

\subsection{Convertible codes~\texorpdfstring{\cite{maturana2020convertible,maturana2020access}}{}}\label{sec:background:convertible}
Convertible codes are erasure codes which are designed to enable encoded data to undergo efficient conversion.
Let $\ICode$ be an $[\In, \Ik]$ code over $\Field{q}$, and $\FCode$ be an $[\Fn, \Fk]$ code over $\Field{q}$.
Previous works on \codenames (and also the present paper) focus on the case where both $\ICode$ and $\FCode$ are linear codes.
In the initial configuration, data will be encoded under the \emph{initial code} $\ICode$, and in the final configuration data will be encoded under the \emph{final code} $\FCode$.
Let $\Ir = (\In - \Ik)$ and $\Fr = (\Fn - \Fk)$.
In order to allow for a change in code dimension from $\Ik$ to $\Fk$, multiple \stripes of codes $\ICode$ and $\FCode$ are considered.
The reason behind this is that in the initial and final configurations, the system must encode the same total number of message symbols (though encoded differently).
Thus, even the simplest instance of the problem involves multiple \stripes in the initial and final configuration.
Let $\Msg$ be a message of length $M = \Lcm(\Ik, \Fk)$ which in the initial configuration is encoded as $\Is = (\sfrac{M}{\Ik})$ \stripes of $\ICode$ and in the final configuration is encoded as $\Fs = (\sfrac{M}{\Fk})$ \stripes of $\FCode$.
For a subset $\Indices \subseteq [M]$, we denote the restriction of $\Msg$ to the coordinates in $\Indices$ as $\Msg|_{\Indices} \in \Field{q}^{|\Indices|}$.
The mapping of message symbols from $\Msg$ to different \stripes is specified by two partitions of $[M]$: an initial partition $\IPart$ and a final partition $\FPart$.
Each subset $\Iset{i} \in \IPart$ must be of size $|\Iset{i}| = \Ik$, and indicates that the submessage $\Msg|_{\Iset{i}}$ is encoded by initial \stripe $i$, for $i \in [\Is]$.
Similarly, each subset $\Fset{j} \in \IPart$ must be of size $|\Fset{j}| = \Fk$, and indicates that the submessage $\Msg|_{\Fset{j}}$ is encoded by final \stripe $j$, for $j \in [\Fs]$.
A conversion from initial code $\ICode$ to final $\FCode$ is a procedure that takes the initial \stripes $\{\ICode(\Msg|_\Iset{i}) : i \in [\Is] \}$ and outputs the final \stripes $\{\FCode(\Msg|_\Fset{i}) : i \in [\Fs] \}$.
Putting all these elements together, a convertible code is formally defined as follows.

\begin{definition}[Convertible code~\cite{maturana2020convertible}]
    \label{def:convertible-code}
    An \ParamCodeDefault\ over $\Field{q}$ is defined by : 
    (1) a pair of initial and final codes $(\ICode, \FCode)$ over $\Field{q}$, where $\ICode$ is an $[\In, \Ik]$ code and $\FCode$ is an $[\Fn, \Fk]$ code, 
    (2) initial and final partitions $(\IPart, \FPart)$ of $\NNodes$ such that $|\Iset{i}| = \Ik$ for $\Iset{i} \in \IPart$ and $|\Fset{j}| = \Fk$ for $\Fset{j} \in \FPart$, 
    (3) a conversion procedure from $\ICode$ to $\FCode$.
\end{definition}

The \accesscost\ of a conversion procedure is the sum of the \emph{read access cost}, i.e.\ the total number of code symbols read, and the \emph{write access cost}, i.e.\ the total number of code symbols written.
An \emph{access-optimal convertible code} is a convertible code whose conversion procedure has the minimum access cost over all convertible codes with given parameters \((\In, \Ik;\allowbreak \Fn, \Fk)\).
Similarly, an $[\In, \Ik]$ code is said to be \accessConvertibleDefault if it is the initial code of an \accessoptimal \ParamCodeDefault.

\Cref{def:convertible-code} considers single fixed values for parameters $\Fn$ and $\Fk$.
In practice, the values of $\Fn$ and $\Fk$ for the conversion might be unknown.
Thus, constructing \codenames which are simultaneously \accessConvertibleDefault for several possible values of $\Fn$ and $\Fk$ is also important (as will be discussed in \cref{sec:merge-bw:construction:multi}).

Though the definition of convertible codes allows for any kind of initial and final codes, in this work we focus exclusively on erasure codes that are MDS.
We call an \ParamCodeDefault\ MDS when both $\ICode$ and $\FCode$ are MDS.
The access cost lower bound for linear MDS convertible codes is known.
\begin{theorem}[\hspace{1sp}\cite{maturana2020access}]
    \label{thm:access-bound}
    Let \(d_1\) be the read access cost of a linear MDS \ParamCodeDefault, and $d_2$ its write access cost..
    When $\Ik \neq \Fk$,
    for every access-optimal code:
    \[
        d_1 \geq
        \begin{cases}
            \Is\Fr + [\Is \bmod \Fs](\Ik - \max\{[\Fk \bmod \Ik], \Fr\}), &
            \text{if } \Ir \geq \Fr \text{ and } \Fr < \min\{\Ik, \Fk\} \\
            \NNodes, &
            \text{otherwise}
        \end{cases}
    \]
    \[
        d_2 \geq \Fs\Fr.
    \]
\end{theorem}
There are explicit constructions~\cite{maturana2020convertible,maturana2020access} of \accessoptimal \codenames for all valid parameters \((\In, \Ik;\allowbreak \Fn, \Fk)\).
Notice that for \regimeone ($\Ir < \Fr$), read access cost is always $\NNodes$, which is the same as the default approach.
In \regimetwo ($\Ir \geq \Fr$), on the other hand, one can achieve lower access cost than the default approach when $\Fr < \min\{\Ik, \Fk\}$.

During conversion, code symbols from the initial \stripes can play multiple roles: they can become part of different final \stripes, their contents might be read or written, additional code symbols may be added and existing code symbols may be removed.
Based on their role, code symbols can be divided into three groups: (1) \emph{unchanged symbols}, which are present both in the initial and final \stripes without any modifications; (2) \emph{retired symbols}, which are only present in the initial \stripes but not in the final \stripes; and (3) \emph{new symbols}, which are present only in the final \stripes but not in the initial \stripes.
Clearly, both unchanged and retired symbols may be read during conversion, and then linear combinations of data read are written into the new symbols.
Convertible codes which have the maximum number of unchanged symbols ($M$ when $\Ik \neq \Fk$) are called \emph{stable}.

The \emph{merge regime} is a fundamental regime of convertible codes which corresponds to conversions which merge multiple initial \stripes into a single final \stripe.
Thus, convertible codes in the merge regime are such that $\Fk = \Cs\Ik$ for some integer $\Cs \geq 2$.
We recall two lemmas from previous work which are useful for analyzing the merge regime.
\begin{proposition}[\hspace{1sp}\cite{maturana2020convertible}]
    \label{thm:no-partitions}
    For every \ParamCombCodeDefault, all possible pairs of initial and final partitions $(\IPart, \FPart)$ are equivalent up to relabeling.
\end{proposition}
In the merge regime, all data gets mapped to the same final stripe.
Thus, the initial and final partition do not play an important role in this case.
\begin{proposition}[\hspace{1sp}\cite{maturana2020convertible}]
    \label{thm:max-unchanged}
    In an MDS \ParamCombCodeDefault, there can be at most $\Ik$ unchanged symbols from each initial \stripe.
\end{proposition}
This is because having more than $\Ik$ unchanged symbols in an initial \stripe would contradict the MDS property.

\textbf{Access optimal convertible code for merge regime.}
In the merge regime, the bound from \cref{thm:access-bound} in the case where $\Ir \geq \Fr$ and $\Fr < \Ik$ reduces to $d_1 \geq \Cs\Fr$ and $d_2 \geq \Fr$.
Thus in \accessoptimal conversion in the merge regime, only $\Fr$ code symbols from each initial \stripe need to be read.
These symbols are then used to compute $\Fr$ new code symbols.

In \cite{maturana2020convertible}, a simple construction for \accessoptimal \codenames in the merge regime is proposed.
Codes built using this construction are (1) systematic,  (2) linear, (3) during conversion only access the first $\Fr$ parities from each initial stripe (assuming $\Fr \leq \Ir$), and (4) when constructed with a given value of $\Is = \Cs$ and $\Fr = r$, the initial $[\In, \Ik]$ code is \accessConvertibleDefault for all $\Fk = \Cs'\Ik$ and $\Fn = \Fk + r'$ such that $1 \leq \Cs' \leq \Cs$ and $1 \leq r' \leq r$.
In \cref{sec:merge-bw:construction} we use an access-optimal \codename in the merge regime as part of our construction of bandwidth-optimal convertible codes for the merge regime.
We will assume, without loss of generality, that the code has these four properties.

\subsection{Network information flow}
\label{sec:background:regenerating}

\emph{Network information flow}~\cite{ahlswede2000network} is a class of problems that model the transmission of information from sources to sinks in a point-to-point communication network.
\emph{Network coding}~\cite{li2003linear,koetter2003algebraic,ho2006random,sanders2003polynomial,jaggi2005polynomial} is a generalization of store-and-forward routing, where each node in the network is allowed to combine its inputs using a code before communicating messages to other nodes.
For the purposes of this paper, an \emph{information flow graph} is a directed acyclic graph with $G = (V, E)$, where $E \subseteq V \times V \times \mathbb{R}_{\geq 0}$ is the set of edges with non-negative capacities, and $(i, j, c) \in E$ represents that information can be sent noiselessly from node $i$ to node $j$ at rate $c$.
Let $X_1, X_2, \ldots, X_m$ be mutually independent information sources with rates $x_1, x_2, \ldots, x_m$ respectively.
Each information source $X_i$ is associated with a source $s_i \in V$, where it is generated, and a sink $t_i \in V$, where it is required.
In this paper we mainly make use of the \emph{information max-flow bound}~\cite{yeung2002multi} which indicates that it is impossible to transmit $X_i$ at a higher rate than the maximum flow from $s_i$ to $t_i$.
In other words, $x_i \leq \operatorname{max-flow}(s_i, t_i)$ for all $i \in [m]$ is a necessary condition for a network coding scheme satisfying all constraints to exist.
In our analysis, we will consider $s_i$-$t_i$-cuts of the information flow graph, which give an upper bound on $\operatorname{max-flow}(s_i, t_i)$ and thus an upper bound on $x_i$ as well.
We will also utilize the fact that two independent information sources with the same source and sink can be considered as a single information source with rate equal to the sum of their rates.

In~\cite{dimakis2010network}, information flow and network coding is applied to the \emph{repair problem} in distributed storage systems.
The repair problem is the problem of reconstructing a small number of failed code symbols in an erasure code (without having to decode the full \stripe).
Dimakis et al.~\cite{dimakis2010network} use information flow to establish bounds on the storage size and repair network-bandwidth of erasure codes.
Similarly, in this work we use information flow to model the process of code conversion and establish lower bounds on the total amount of network bandwidth used during conversion.

\begin{figure}
    \centering
    \arraycolsep=4pt 
    \renewcommand{\arraystretch}{1.3}
    \begin{equation*}
        \footnotesize
        \begin{array}{c|c|c|c|c|c|c|c|c|c|}
            \cline{2-5}\cline{7-10}
            \text{Symbol } 1 &
            f_1(\Msg_1) &
            f_1(\Msg_2) &
            \cdots &
            f_1(\Msg_\alpha) & &
            f_1(\Msg_1) &
            f_1(\Msg_1) + g_{2,1}(\Msg_2) &
            \cdots &
            f_1(\Msg_\alpha) + g_{\alpha,1}(\Msg_1, \ldots, \Msg_\alpha)
            \\\cline{2-5}\cline{7-10}
            \vdots &
            \vdots &
            \vdots &
            \ddots &
            \vdots & &
            \vdots &
            \vdots &
            \ddots &
            \vdots
            \\\cline{2-5}\cline{7-10}
            \text{Symbol } n &
            f_n(\Msg_1) &
            f_n(\Msg_2) &
            \cdots &
            f_n(\Msg_\alpha) & &
            f_n(\Msg_1) &
            f_n(\Msg_1) + g_{2,n}(\Msg_2) &
            \cdots &
            f_1(\Msg_\alpha) + g_{\alpha,n}(\Msg_1, \ldots, \Msg_\alpha)
            \\\cline{2-5}\cline{7-10}
            \multicolumn{10}{c}{}
            \\[-.9em]
            \multicolumn{1}{c}{} &
            \multicolumn{4}{c}{\text{(a) $\alpha$ instances of the base code}} &
            \multicolumn{1}{c}{} &
            \multicolumn{4}{c}{\text{(b) Piggybacked code}}
            \\
        \end{array}
    \end{equation*}
    {\phantomsubcaption\label{fig:piggyback:base}}
    {\phantomsubcaption\label{fig:piggyback:piggy}}
    \caption{\piggyfw~\cite{rashmi2017piggybacking} for constructing vector codes.}
    \label{fig:piggyback}
\end{figure}

\subsection{\piggyfw for constructing vector codes}
\label{sec:background:piggyback}
The \emph{\piggyfw}~\cite{rashmi2013piggybacking,rashmi2017piggybacking} is a framework for constructing new vector codes building on top of existing codes.
The main technique behind the \piggyfw is to take an existing code as a \emph{base code}, create a new vector code consisting of multiple instances of the base code (as described below), and then add carefully designed functions of the data (called \emph{piggybacks}) from one instance to the others.
These piggybacks are added in a way such that it retains the decodability properties of the base code (such as the MDS property).
The piggyback functions are chosen to confer additional desired properties to the resulting code.
In~\cite{rashmi2017piggybacking}, the authors showcase the \piggyfw by constructing codes that are efficient in reducing bandwidth consumed in repairing \stripe symbols.

More specifically, the \piggyfw works as follows.
Consider a length $n$ code defined by encoding function $f(\Msg) = (f_1(\Msg), f_2(\Msg), \ldots, f_n(\Msg))$.
Now, consider $\alpha$ instances of this base code, each corresponding to a coordinate of the $\alpha$-length vector of each symbol in the new vector code.
Let $\Msg_1, \Msg_2, \cdots, \Msg_\alpha$ denote the independent messages encoded under these $\alpha$ instances, as shown in \cref{fig:piggyback:base}.
For every $i$ such that $2 \leq i \leq \alpha$, one can add to the data encoded in instance $i$ an arbitrary function of the data encoded by instances $\{1, \ldots, (i - 1)\}$.
Such functions are called \emph{piggyback functions}, and the piggyback function corresponding to code symbol $j \in [n]$ of instance $i \in \{2, \ldots, \alpha\}$ is denoted as $g_{i,j}$.

The decoding of the piggybacked code proceeds as follows.
Observe that instance $1$ does not have any piggybacks. 
First, instance $1$ of the base code is decoded using the base code's decoding procedure in order to obtain $\Msg_1$.
Then, $\Msg_1$ is used to compute and subtract any of the piggybacks $\{g_{2,i}(\Msg_1)\}_{i=1}^n$ from instance $2$ and the base code's decoding can then be used to recover $\Msg_2$.
Decoding proceeds like this, using the data decoded from previous instances in order to remove the piggybacks until all instances have been decoded.
It is clear that if an $[n, k, \alpha]$ vector code is constructed from an $[n, k]$ MDS code as the base code using the \piggyfw, then the resulting vector code is also MDS.
This is because any set of $k$ symbols from the vector code contains a set of $k$ subsymbols from each of the $\alpha$ instances.

In this paper, we use the \piggyfw to design a code where piggybacks store data which helps in making the conversion process efficient.

%% file: sections/3-related_work.tex
\subsection{Other related work}
\label{sec:related-work}

Apart from~\cite{maturana2020convertible}, which presented a general formulation for the code conversion problem, special cases of code conversion have been studied in the literature.
In~\cite{xia2015tale}, the authors propose two specific pairs of non-MDS codes for a distributed storage system which support conversion with lower access cost than the default approach.
In~\cite{su2020local}, the authors study two kinds of conversion in the context of distributed matrix multiplication.
These works focus on reducing the access cost of conversion, whereas the focus of the current paper is on the bandwidth cost of conversion.
Furthermore, the approaches proposed in these works~\cite{xia2015tale,su2020local} do not come with any theoretical guarantees on optimality, whereas the current paper also presents tight lower bounds on the bandwidth cost of conversion along with bandwidth-optimal constructions.

A related line of research is that of regenerating codes.
Regenerating codes are erasure codes which are designed to solve the repair problem (described in \cref{sec:background:regenerating} above) by downloading the least amount of data from the surviving nodes.
Regenerating codes were first proposed by Dimakis et al.\ \cite{dimakis2010network}.
Several subsequent works~(\eg~\cite{rashmi2011optimal,shah2011distributed, shah2012interference, suh2011exact,cadambe2011polynomial,wang2011codes, tamo2013zigzag,papailiopoulos2013repair,alrabiah2019exponential,balaji2018tight,chowdhury2018new,mahdaviani2018product,rashmi2017piggybacking,sasidharan2017explicit,ye2017explicita,rawat2018mds,goparaju2017minimum,rashmi2014hitchhikers,tamo2014access,rashmi2013solution,cadambe2013asymptotic,wang2012long,shum2011cooperative,rashmi2011enabling,shah2010flexible,ye2016explicit,dau2018repairing,guruswami2017repairing,li2018generic,mardia2019repairing,shanmugam2014repair,tamo2017optimal,kamath2014codes} and references therein)
have provided constructions and generalizations of regenerating codes.
The regenerating codes framework measures the cost of repair in a similar way to how we measure the cost of conversion in this work: in terms of the total amount of network bandwidth used, i.e.\ the total amount of data transferred during repair. 
Thus, some of the techniques used in this paper are inspired by the existing regenerating codes literature, as further explained in \cref{sec:background:regenerating}.
Furthermore, specific instances of code conversion can be viewed as instances of the repair problem, for example, increasing $n$ while keeping $k$ fixed as studied in~\cite{rashmi2011enabling,rashmi2017piggybacking,mousavi2018delayed}.
In such a scenario, one can view adding additional nodes as ``repairing'' them as proposed in \cite{rashmi2011enabling}.
Note that this setting imposes a relaxed requirement of repairing only a specific subset of nodes as compared to regenerating codes which require optimal repair of all nodes.
Yet, the lower bound from regenerating codes still applies for MDS codes, since as shown in~\cite{shah2012interference}, the regenerating codes lower bound for MDS codes applies even for repair of only a single specific node.

Another related line of research is that of locally recoverable codes, also known as local reconstruction codes, or LRCs for short.
LRCs are non-MDS codes with the property that any \stripe symbol can be recovered by reading a relatively small subset of other symbols (and usually much smaller than the subset of symbols required to decode the full data).
Several works~(\eg \cite{gopalan2012locality,gopalan2014explicit,papailiopoulos2014locally,tamo2014family,kamath2014codes,cadambe2015bounds,tamo2016optimal,barg2017locally,frankfischer2017locality,mazumdar2018capacity,guruswami2019how,gopi2019maximally,prakash2012optimal} and references therein) have studied the properties of LRCs (and variants thereof) and proposed constructions.
While the cost metric of LRCs more closely resembles the access cost metric, the constraint that each initial and final \stripe in a convertible code can be decoded independently may be seen as a form of local decodability.

There have been several works studying the {scaling problem}~\cite{zhang2010alv,zheng2011fastscale,wu2012gsr,zhang2014rethinking,huang2015scale,wu2016i/o,zhang2018optimal,hu2018generalized,zhang2020efficient,rai2015adaptivea,rai2015adaptive,wu2020optimal}. 
This problem considers upgrading an erasure-coded storage system with $s$ new empty data nodes.
The general goal is to efficiently and evenly redistribute data across all nodes, while updating parities to reflect the new placement of the data.
This is a fundamentally different problem from the code conversion problem we study in this paper, due to the scaling problem's need to redistribute data across nodes.

%% file: sections/4-model.tex
\section{Modeling conversion for optimizing network bandwidth}
\label{sec:model}

In this section, we model the conversion process as an information flow problem.
We utilize this model primarily for deriving lower bounds on the total amount of information that needs to be transferred during conversion.
Since our focus is on modeling the conversion process, we consider a single value for each of the final parameters $\Fn$ and $\Fk$.
This model continues to be valid for each individual conversion, even when the final parameters might take multiple values.

In \cref{sec:background:convertible}, we reviewed the definition of convertible codes from literature~\cite{maturana2020access,maturana2020convertible}.
Existing works on convertible codes \cite{maturana2020access,maturana2020convertible} have considered only \emph{scalar codes}, where each code symbol corresponds to a scalar from a finite field $\Field{q}$. 
Considering scalar codes is sufficient when optimizing for access cost, which was the focus in these prior works, since the access cost is measured at the granularity of code symbols.
However, when optimizing the cost of network bandwidth, vector codes can perform better than scalar codes since they allow partial download from a node. 
This allows conversion procedures to only download a fraction of a code symbol and thus only incur the bandwidth cost associated with the size of that fraction.
This can potentially lead to significant reduction in network bandwidth cost.
For this reason, we consider the initial code $\ICode$ as an $[\In, \Ik, \alpha]$ MDS code and the final code $\FCode$ as an $[\Fn, \Fk, \alpha]$ MDS code, where $\alpha \geq 1$ is considered as a free parameter chosen to minimize network bandwidth cost.
This move to vector codes is inspired by the work of Dimakis et al.~\cite{dimakis2010network} on regenerating codes, who showed the benefit of vector codes in reducing network bandwidth in the context of the repair problem.
For MDS convertible codes, message size will be $\blocksize = \size\alpha = \Lcm(\Ik, \Fk)\alpha$, which we interpret as a vector $\Msg \in \Field{q}^{\size\alpha}$ composed of $\size$ symbols made up of $\alpha$ subsymbols each.
We will denote the number of symbols downloaded from node $s$ during conversion as $\dl{s} \leq \alpha$ and extend this notation to sets of nodes as $\beta(\mathcal{S}) = \sum_{s \in \mathcal{S}} \dl{s}$.

Consider an \MdsParamCodeDefault with initial partition $\IPart = \{ \Iset{1}, \ldots, \Iset{\Is} \}$ and final partition $\FPart = \{ \Fset{1}, \ldots, \Fset{\Fs} \}$.
We model conversion using an information flow graph as the one shown in \cref{fig:conversion-information-flow} where message symbols are generated at source nodes, and sinks represent the decoding constraints of the final code.
Symbols of message $\Msg$ are modeled as information sources $X_1, X_2, \ldots, X_{\size}$ of rate $\alpha$ (over $\Field{q}$) each.
For each initial \stripe $i \in [\Is]$, we include one source node $\sourcev{i}$, where the information sources corresponding to the message symbols in $\Iset{i}$ are generated.
Each code symbol of initial \stripe $i$ is modeled as a node with an incoming edge from $\sourcev{i}$.
A \emph{coordinator node} $\centralv$ models the central location where the contents of new symbols are computed, and it {has incoming edges from} all nodes in the initial \stripes.
During conversion, some of the initial code symbols will remain unchanged, some will be retired, and some new code symbols will be added.
Thus, we also include the nodes corresponding to unchanged symbols in the final \stripes (that is, every unchanged node is shown twice in \cref{fig:conversion-information-flow}).
Note that the unchanged nodes in the initial \stripes and the unchanged nodes in the final \stripes are identical, and thus do not add any \bwcost.
For each new symbol we add a node that connects to the coordinator node.
From this point, we will refer to code symbols and their corresponding nodes interchangeably.
For each final \stripe $j \in [\Fs]$, we add a sink $\sinkvs{j}$ which connects to some subset of nodes from final \stripe $j$, and recovers the information sources corresponding to the message symbols in $\Fset{j}$.

\begin{figure}
    \centering
    \includegraphics[]{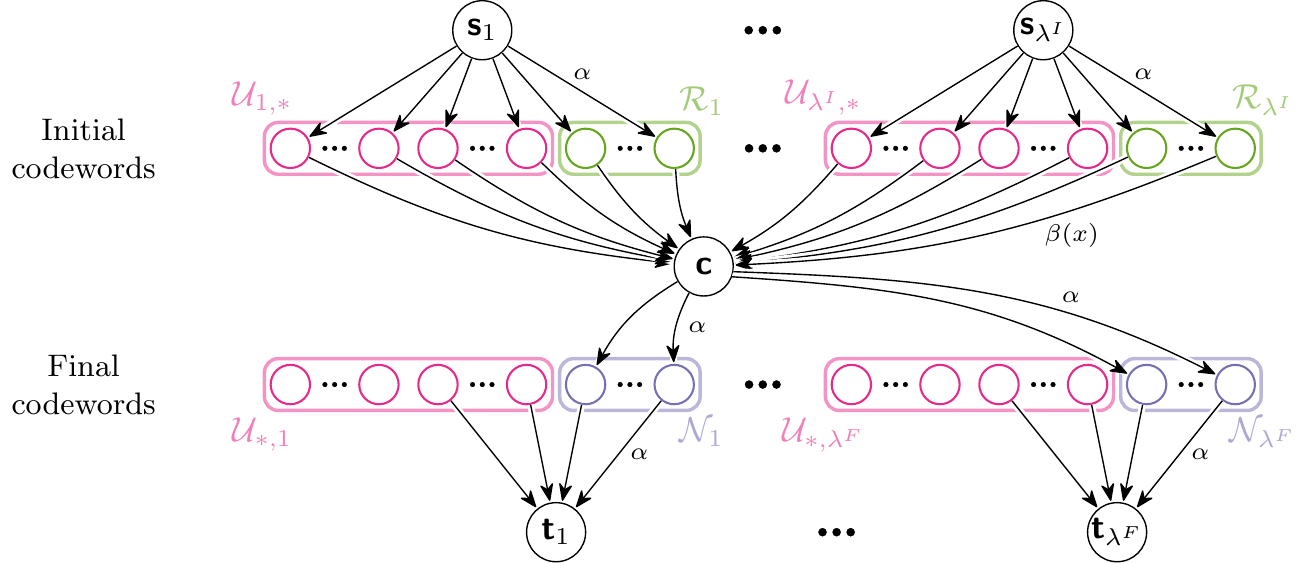}
    \caption{
        Information flow graph of conversion in the general case.
        Unchanged, retired, and new nodes are shown in different colors.
        Notice that each unchanged node in this figure is drawn twice: once in the initial codewords and once in the final codewords.
        These correspond to exactly the same node, but are drawn twice for clarity.
        Some representative edges are labeled with their capacities.
    }
    \label{fig:conversion-information-flow}
\end{figure}

Thus, the information flow graph for a \codename comprises the following nodes:
\begin{itemize}
    \item unchanged nodes $\unchangedvset{i}{j} = \{\unchangedv{i}{j}{1}, \ldots, \unchangedv{i}{j}{|\unchangedvset{i}{j}|}\}$ for all $i \in [\Is],\ j \in [\Fs]$, which are present both in the initial and final \stripes;
    \item retired nodes $\retiredvset{i} = \{\retiredv{i}{1}, \ldots, \retiredv{i}{|\retiredvset{i}|}\}$ for $i \in [\Is]$, which are only present in the initial \stripes;
    \item new nodes $\newvsets{j} = \{\newvs{j}{1}, \ldots, \newvs{j}{|\newvsets{j}|}\}$ for $j \in [\Fs]$, which are only present in the final \stripes;
    \item source nodes $\sourcev{i}$ for $i \in [\Is]$, representing the data to be encoded;
    \item sink nodes $\sinkvs{j}$ for $j \in [\Fs]$, representing the data decoded; and
    \item a coordinator node $\centralv$.
\end{itemize}
In the information flow graph, information source $X_{l}$ is generated at node $\sourcev{i}$ if and only if $l \in \Iset{i}$, and recovered at node $\sinkvs{l}$ if and only if $l \in \Fset{j}$.

Throughout this paper, we use the disjoint union symbol $\sqcup$ when appropriate to emphasize that the two sets in the union are disjoint.
To simplify the notation, when $*$ is used as an index, it denotes the disjoint union of the indexed set over the range of that index, e.g.\ $\unchangedvset{*}{j} = \bigsqcup_{i=1}^{\Is} \unchangedvset{i}{j}$.

The information flow graph must be such that the following conditions hold: (1) the number of nodes per initial \stripe is $\In$, i.e., $|\unchangedvset{i}{*}| + |\retiredvset{i}| = \In$ for all $i \in [\Is]$; and (2) the number of nodes per final \stripe is $\Fn$, i.e., $|\unchangedvset{*}{j}| + |\newvsets{j}| = \Fn$ for all $j \in [\Fs]$.
Additionally, the information flow graph contains the following set of edges $E$, where a directed edge from node $u$ to $v$ with capacity $\delta$ is represented with the triple $(u, v, \delta)$:
\begin{itemize}
    \item $\{(\sourcev{i}, x, \alpha) : x \in \unchangedvset{i}{*} \sqcup \retiredvset{i}\} \subset E$ for each $i \in [\Is]$, {where the capacity corresponds to the size of the data stored on each node;}
    \item $\{(x, \centralv, \dl{x}) : x \in \unchangedvset{i}{*} \sqcup \retiredvset{i}\} \subset E$ for each $i \in [\Is]$, {where the capacity corresponds to the amount of data downloaded from node $x$;}
    \item $\{(\centralv, y, \alpha) : y \in \newvsets{j}\} \subset E$ for each $j \in [\Fs]$, {where the capacity corresponds to the size of the data stored on each new node;}
    \item $\{(y, \sinkvs{j}, \alpha) : y \in V_j\} \subset E$ for $V_j \subseteq \unchangedvset{*}{j} \sqcup \newvsets{j}$ such that $|V_j| = \Fk$, for all $j \in [\Fs]$, {where the capacity corresponds to the size of the data stored on each node}.
\end{itemize}
The sinks $\sinkvs{j}$ represent the decoding constraints of the final code, and each choice of set $V_j$ will represent a different choice $k$ code symbols for decoding the final \stripe.
A necessary condition for a conversion procedure is to satisfy all sinks $\sinkvs{j}$ for all possible $V_1, \ldots, V_{\Fs}$.
The sets $\unchangedvset{i}{j}, \retiredvset{i}, \newvsets{j}$ and the capacities $\dl{x}$ are determined by the conversion procedure of the convertible code.
\Cref{fig:conversion-information-flow} shows the information flow graph of an arbitrary convertible code.

\begin{definition}[Conversion bandwidth]
    The \emph{conversion bandwidth} $\gamma$ is the total network bandwidth used during conversion and is equal to the total amount of data that is transferred to the coordinator node $\centralv$ from the initial nodes plus the total amount of data transferred to the new nodes from the coordinator node $\centralv$, that is:
    \begin{equation}\label{eq:bandwidth}
        \gamma = \dl{\unchangedvset{*}{*} \sqcup \retiredvset{*}} + |\newvsets{*}|\alpha.
    \end{equation}
\end{definition}
Once the structure of the graph is set and fixed, information flow analysis gives lower bounds on the capacities $\dl{x}$.
Therefore, a part of our objective in designing convertible codes is to set $\unchangedvset{i}{j}, \retiredvset{i}, \newvsets{j}$ so as to minimize the lower bound on $\gamma$.

\begin{remark}
    In practice, \convbw can sometimes be further reduced by placing the coordinator node along with a new node and/or a retired node in the same server.
    One can even first split the coordinator node into several coordinator nodes, each processing data which is not used in conjunction with data processed by other coordinator nodes, and then place them in the same server as a new node and/or a retired node.
    Such ``optimizations'' do not fundamentally alter our result, and hence are left out in order to make the exposition clear.
\end{remark}

%% file: sections/5-merge_regime_lower_bound.tex
\section{Optimizing network bandwidth of conversion in the merge regime}
\label{sec:merge-bw}
\label{sec:merge-bw:lower-bound}

\begin{figure}
    \begin{subfigure}{.495\textwidth}
    \centering
    \includegraphics[width=.95\textwidth]{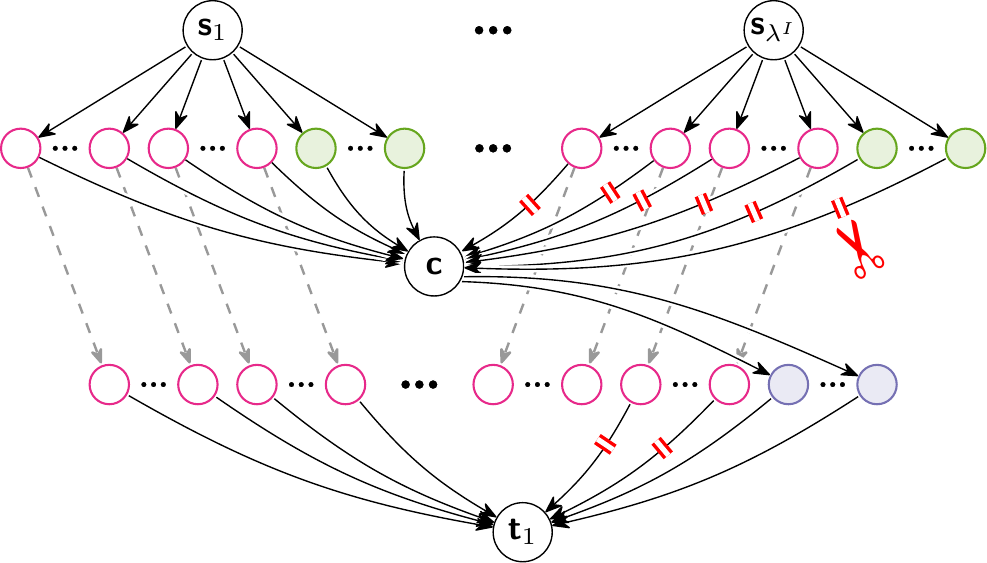}
    \caption{}
    \label{fig:merge-flow:cut-1}
    \end{subfigure}
    \begin{subfigure}{.495\textwidth}
    \centering
    \includegraphics[width=.95\textwidth]{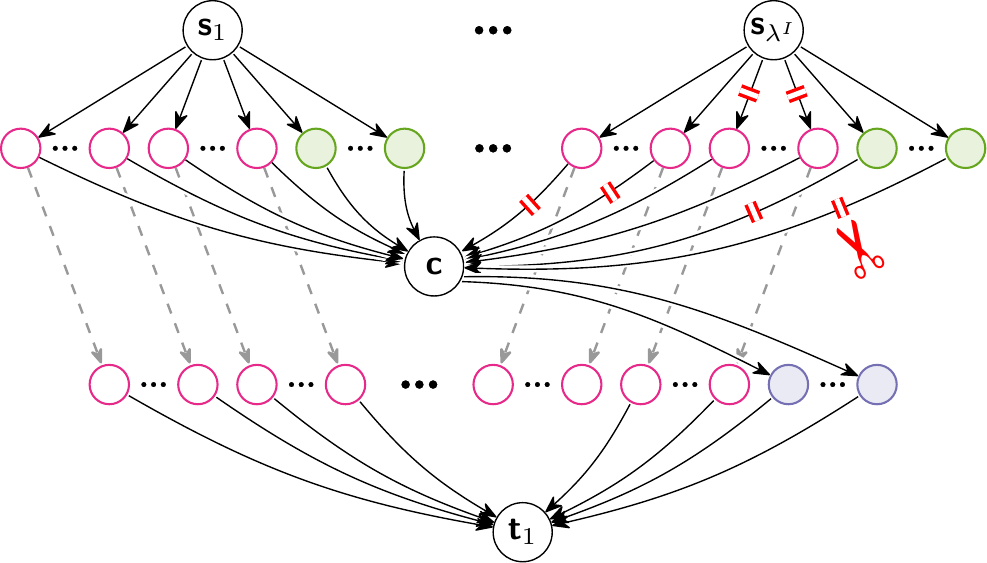}
    \caption{}
    \label{fig:merge-flow:cut-2}
    \end{subfigure}
    \caption{
        Information flow graph of conversion in the merge regime with two different cuts (used in proofs).
        For clarity, each unchanged node is drawn twice: once in the initial \stripes and once in the final \stripe.
        These two instances are connected by a dashed arrow.
        Marked edges denote a graph cut.
    }
    \label{fig:merge-flow}
\end{figure}

In this section, we use the information flow model presented in \cref{sec:model} to derive a lower bound on the \convbw for MDS codes in the merge regime. 
Recall from \cref{sec:background:convertible}, that convertible codes in the merge regime are those where $\Fk = \Cs\Ik$ for some integer $\Cs \geq 2$, i.e., this regime corresponds to conversions were multiple initial \stripes are merged into a single final \stripe.
As in the previous section, our analysis focuses on a single conversion, and thus a single value for the final parameters $\Fn$ and $\Fk$.
The lower bound on conversion bandwidth derived in this section continues to hold even when we consider multiple possible values for the final parameters $\Fn$ and $\Fk$.

Consider an \ParamCombCodeDefault in the merge regime, for some integer $\Cs \geq 2$.
Note that for all convertible codes in the merge regime, it holds that $\Is = \Cs$ and $\Fs = 1$.
Since all initial and final partitions $(\IPart, \FPart)$ are equivalent up to relabeling in this regime (by \cref{thm:no-partitions}~\cite{maturana2020convertible}), we can omit them from our analysis.
Note also that every information source shares the same sink, as there is only a single sink $\sinkvs{1}$.
Thus, we may treat each source $s_i$ as having a single information source $X_i$ of rate $\alpha \Ik$ ($i \in [\Is]$).
\Cref{fig:merge-flow:cut-1} shows the information flow graph for a convertible code in the merge regime.

First, we derive a general lower bound on \convbw in the merge regime by considering a simple cut in the information flow graph.
Intuitively, this lower bound emerges from the fact that new nodes need to have a certain amount of information from each initial \stripe in order to fulfill the MDS property of the final code. 
This lower bound depends on the number of unchanged nodes and achieves its minimum when the number of unchanged nodes is maximized.
Recall from \cref{sec:background:convertible} that \codenames with maximum number of unchanged nodes are called \textit{stable} \codenames. Thus, the derived lower bound is minimized for stable \codenames.
\begin{lemma}\label{thm:merge-bw:lower-bound:one}
    Consider an \mds\ \ParamCombCodeDefault.
    Then \(\gamma \geq \Cs\alpha\min\{\Fr, \Ik\} + \Fr\alpha\), where equality is only possible for stable codes.
\end{lemma}
\begin{proof}
    We prove this inequality via an information flow argument.
    Let $i \in [\Is]$ and consider the information source generated at source $\sourcev{i}$.
    Let $S \subseteq \unchangedvset{i}{1}$ be a subset of unchanged nodes from initial \stripe $i$ of size $\Ssize{i} = \min\{\Fr, |\unchangedvset{i}{1}|\}$.
    Consider a sink $\sinkvs{1}$ that connects to nodes $\unchangedvset{*}{1} \setminus S$.
    We choose the graph cut defined by nodes $\{\sourcev{i}\} \sqcup \unchangedvset{i}{1} \sqcup \retiredvset{i}$ (see \cref{fig:merge-flow:cut-1}, which depicts the cut for $i = \Is$).
    This cut yields the following inequality:
    \[
        \Ik\alpha \leq \max\{|\unchangedvset{i}{1}| - \Fr, 0\}\alpha + \dl{\unchangedvset{i}{1} \sqcup \retiredvset{i}}
    \]
    \[
        \iff \dl{\unchangedvset{i}{1} \sqcup \retiredvset{i}} \geq
        (\Ik + \Fr - \max\{|\unchangedvset{i}{1}|, \Fr\})\alpha
    \]
    By summing this inequality over all sources $i \in [\Is]$ and using the definition of $\gamma$ (\cref{eq:bandwidth}), we obtain:
    \[
        \gamma \geq \sum_{i=1}^{\Is} (\Ik + \Fr - \max\{|\unchangedvset{i}{1}|, \Fr\})\alpha + |\newvsets{1}|\alpha
    \]
    By \cref{thm:max-unchanged}\cite{maturana2020convertible}, $|\unchangedvset{i}{1}| \leq \Ik$.
    Therefore, it is clear that the right hand side achieves its minimum if and only $|\unchangedvset{i}{1}| = \Ik$ for all $i \in [\Is]$, proving the result.
\end{proof}
\begin{remark}
    Note that the \convbw lower bound described in \cref{thm:merge-bw:lower-bound:one} coincides with the access-cost lower bound described in \cref{thm:access-bound} when $\Ir \geq \Fr$.
    This follows by recalling that each node corresponds to an $\alpha$-length vector, and for scalar codes $\alpha = 1$.
\end{remark}

In particular, this implies that convertible codes in the merge regime which are access-optimal and have $\Ir \geq \Fr$ are also bandwidth-optimal.
Observe that this corresponds to \regimeone.
However, as we will show next, this property fails to hold when $\Ir < \Fr$ (that is, \regimetwo).

We next derive a lower bound on \convbw which is tighter than \cref{thm:merge-bw:lower-bound:one} when $\Ir < \Fr$.
Nevertheless, it allows for less \convbw usage than the \accessoptimal codes.

Intuitively, the data downloaded from retired nodes during conversion will be ``more useful'' than the data downloaded from unchanged nodes, since unchanged nodes already form part of the final \stripe.
At the same time, it is better to have the maximum amount of unchanged nodes per initial \stripe ($\Ik$) because this minimizes the number of new nodes that need to be constructed.
However, this leads to fewer retired nodes per initial \stripe ($\Ir$).
If the number of retired nodes per initial \stripe is less than the number of new nodes ($\Ir < \Fr$), then conversion procedures are forced to download data from unchanged nodes.
This is because one needs to download at least $\Fr\alpha$ from each initial \stripe (by \cref{thm:merge-bw:lower-bound:one}).
Since data from unchanged nodes is ``less useful'', more data needs to be downloaded in order to construct the new nodes.

As in the case of \cref{thm:merge-bw:lower-bound:one}, this lower bound depends on the number of unchanged nodes in each initial \stripe, and achieves its minimum in the case of stable convertible codes.
\begin{lemma}\label{thm:merge-bw:lower-bound:two}
    Consider an \mds\ \ParamCombCodeDefault, with \(\Ir < \Fr \leq \Ik\).
    Then \(\gamma \geq \Cs\alpha\left(\Ir + \Ik\left( 1 - \frac{\Ir}{\Fr} \right) \right) + \Fr\alpha\), where equality is only possible for stable codes.
\end{lemma}
\begin{proof}
    We prove this via an information flow argument.
    Let $i \in [\Is]$ and consider the information source generated at source $\sourcev{i}$.
    Let $S \subseteq \unchangedvset{i}{1}$ be a subset of size $\Ssize{i} = \min\{\Fr, |\unchangedvset{i}{1}|\}$.
    Consider a sink $\sinkvs{1}$ that connects to the nodes in $\unchangedvset{*}{1} \setminus S$.
    We choose the graph cut defined by nodes $\{\sourcev{i}\} \sqcup S \sqcup \retiredvset{i}$ (see \cref{fig:merge-flow:cut-2}, which depicts the cut when $i = \Is$).
    This yields the following inequality:
    \[
        \Ik\alpha \leq (|\unchangedvset{i}{1}| - \Ssize{i})\alpha
        + \dl{S} + \dl{\retiredvset{i}}.
    \]
    By rearranging this inequality and summing over all possible choices of subset $S$, we obtain the following inequality:
    \begin{equation*}
        \binom{|\unchangedvset{i}{1}|}{\Ssize{i}}
        (\Ik + \Ssize{i} - |\unchangedvset{i}{1}|)\alpha \leq
        \binom{|\unchangedvset{i}{1}| - 1}{\Ssize{i} - 1}
        \dl{\unchangedvset{i}{1}} + 
        \binom{|\unchangedvset{i}{1}|}{\Ssize{i}}
        \dl{\retiredvset{i}}
    \end{equation*}
    \begin{equation} \label[ineq]{eq:merge-flow}
        \iff |\unchangedvset{i}{1}|
        (\Ik + \Ssize{i} - |\unchangedvset{i}{1}|)\alpha \leq
        \Ssize{i}
        \dl{\unchangedvset{i}{1}} + 
        |\unchangedvset{i}{1}|
        \dl{\retiredvset{i}}.
    \end{equation}
    Then, our strategy to obtain a lower bound is to find the minimum value for conversion bandwidth $\gamma$ which satisfies \cref{eq:merge-flow} for all $i \in [\Is]$, which can be formulated as the following optimization problem:
    \begin{equation} \label[linearprogram]{eq:merge-lp}
    \begin{array}{rl}
        \text{minimize} & \gamma =
        \sum_{i \in \Is} \left[\dl{\unchangedvset{i}{1}} + \dl{\retiredvset{i}}\right]
        + |\newvsets{1}|\alpha \\
        \text{subject to} & \text{\cref{eq:merge-flow}, for all $i \in [\Is]$} \\
                          & 0 \leq \dl{x} \leq \alpha, \text{ for all } x \in \unchangedvset{*}{1} \sqcup \retiredvset{*}.
    \end{array}
    \end{equation}
    Intuitively, this linear program shows that it is preferable to download more data from retired nodes ($\dl{\retiredvset{i}}$) than unchanged nodes ($\dl{\unchangedvset{i}{1}}$), since both have the same impact on $\gamma$ but the contribution $\dl{\retiredvset{i}}$ towards satisfying \cref{eq:merge-flow} is greater than or equal than that of $\dl{\retiredvset{i}}$, because $\Ssize{i} \leq |\unchangedvset{i}{1}|$ by definition.
    Thus to obtain an optimal solution we first set $\dl{\retiredvset{i}} = \min\{\Ik + \Ssize{i} - |\unchangedvset{i}{1}|, |\retiredvset{i}|\}\alpha$ to the maximum needed for all $i \in [\Is]$, and then set:
    \[
        \sum_{x \in \unchangedvset{i}{1}} \dl{x}
        =
        \frac{
        \max\{\Ssize{i} - \Ir, 0\}
        |\unchangedvset{i}{1}|\alpha
        }
        {\Ssize{i}},
        \qquad
        \text{for all } i \in [\Is]
    \]
    to satisfy the constraints.
    It is straightforward to check that this solution satisfies the KKT (Karush-Kuhn-Tucker) conditions, and thus is an optimal solution to \cref{eq:merge-lp}.
    By replacing these terms back into $\gamma$ and simplifying we obtain the optimal objective value:
    \[
        \gamma^* =
        \sum_{i=1}^{\Is}
        \left[
        \Ik
        -
        \min\{\Ir, \Ssize{i}\}
        \left( \frac{|\unchangedvset{i}{1}|}{\Ssize{i}} - 1 \right)
        \right]\alpha
        +
        |\newvsets{1}|\alpha
    \]
    It is easy to show that the right hand side achieves its minimum if and only if $|\unchangedvset{i}{1}| = \Ik$ for all $i \in [\Is]$ (i.e., the code is stable).
    This gives the following lower bound for conversion bandwidth:
    \[
        \gamma \geq \Is\alpha\left( \Ir + \Ik \left( 1 - \frac{\Ir}{\Fr} \right) \right) + \Fr\alpha.
    \]
\end{proof}
By combining \cref{thm:merge-bw:lower-bound:one,thm:merge-bw:lower-bound:two} we obtain the following general lower bound on \convbw of \mds \codenames in the merge regime.
\begin{theorem}\label{thm:merge-bw:lower-bound}
    For any \mds\ \ParamCombCodeDefault:
    \begin{equation*}
        \gamma \geq
        \begin{cases}
            \Cs\alpha \min\{\Ik, \Fr\} + \Fr\alpha, & \text{if } \Ir \geq \Fr \text{ or } \Ik \leq \Fr\\
            \Cs\alpha\left(\Ir + \Ik \left( 1 - \frac{\Ir}{\Fr} \right) \right) + \Fr\alpha, & \text{otherwise}
        \end{cases}
    \end{equation*}
    where equality can only be achieved by stable \codenames.
\end{theorem}
\begin{proof}
    Follows from \cref{thm:merge-bw:lower-bound:one,thm:merge-bw:lower-bound:two}.
\end{proof}
In \cref{sec:merge-bw:construction}, we show that the lower bound of \cref{thm:merge-bw:lower-bound} is indeed achievable for all parameter values in the merge regime, and thus it is tight.
We will refer to \codenames that meet this bound with equality as \emph{\bandwidthoptimal}.

\begin{remark}
    Observe that the model above allows for nonuniform data download during conversion, that is, it allows the amount of data downloaded from each node during conversion to be different.
    If instead one were to assume uniform download, i.e.\ $\dl{x} = \dl{y}$ for all $x, y \in \unchangedvset{*}{*} \sqcup \retiredvset{*}$, then a higher lower bound for conversion bandwidth $\gamma$ is obtained (mainly due to \cref{eq:merge-flow} in the proof of \cref{thm:merge-bw:lower-bound:two}).
    Since the lower bound of \cref{thm:merge-bw:lower-bound} is achievable, this implies that assuming uniform download necessarily leads to a suboptimal solution.
\end{remark}

\begin{remark}
    The case where $\Ik = \Fk$ can be analyzed using the same techniques used in this section.
    In this case, $\Is = \Cs = 1$.
    There are some differences compared to the case of the merge regime: for example, in this case the number of unchanged nodes can be at most $\min\{\In, \Fn\}$ (in contrast to the $\Cs\Ik$ maximum of the merge regime).
    So, conversion bandwidth in the case where $\In \geq \Fn$ is zero, since we can simply keep $\Fn$ nodes unchanged.
    In the case where $\In < \Fn$, the same analysis from \cref{thm:merge-bw:lower-bound:two} is followed, but the larger number of unchanged nodes will lead to a slightly different inequality.
    Thus, in the case of $\Ik = \Fk$ the lower bound on conversion bandwidth is:
    \[
        \gamma \geq
        \begin{cases}
            0, & \text{if $\In \geq \Fn$}\\
            \alpha\left( \Ik + \Ir \right) \left( 1 - \frac{\Ir}{\Fr}  \right) + (\Fr - \Ir)\alpha, & \text{otherwise.}
        \end{cases}
    \]
    
    Readers familiar with regenerating codes might notice that the above lower bound is equivalent to the lower bound on the repair bandwidth~\cite{dimakis2010network,cadambe2013asymptotic} when $(\Fr - \Ir)$ symbols of an $[\Ik + \Fr, \Ik]$ MDS code are to be repaired with the help of the remaining $(\Ik + \Ir)$ symbols.
    Note that this setting imposes a relaxed requirement of repairing only a specific subset of symbols as compared to regenerating codes which require optimal repair of all nodes.
    Yet, the lower bound remains the same. This is not surprising though, since it has been shown~\cite{shah2012interference} that the regenerating codes lower bound for MDS codes applies even for repair of only a single specific symbol.
\end{remark}

%% file: sections/6-merge_regime_construction.tex
\section{Explicit construction of Bandwidth-optimal MDS convertible codes in\texorpdfstring{\\}{}the merge regime} 
\label{sec:merge-bw:construction}

In this section, we present an explicit construction for \bandwidthoptimal\ \codenames\ in the \combinationregime.
Our construction employs the \piggyfw~\cite{rashmi2017piggybacking}.
Recall from \cref{sec:background:piggyback} that the \piggyfw is a framework for constructing vector codes using an existing code as a base code and adding specially designed functions called piggybacks which impart additional properties to the resulting code.
We use an \accessoptimal \codename to construct the base code and design the piggybacks to help achieve minimum \convbw.
First, in \cref{sec:merge-bw:construction:single}, we describe our construction of \bwoptimal \codenames in the case where we only consider fixed unique values for the final parameters $\Fn$ and $\Fk = \Cs\Ik$.
Then, in \cref{sec:merge-bw:construction:multi}, we show that initial codes built with this construction are not only \bwConvertibleDefault, but also simultaneously bandwidth-optimally convertible for multiple other values of the pair $(\Fn, \Fk)$.
Additionally, we present a construction which given any finite set of possible final parameter values $(\Fn, \Fk)$, constructs an initial $[\In, \Ik]$ code which is simultaneously \bwConvertibleDefault for every $(\Fn, \Fk)$ in that set.

\subsection{Bandwidth-optimal MDS \codenames for fixed final parameters}
\label{sec:merge-bw:construction:single}

The case where $\Fr \geq \Ik$ is trivial, since the default approach to conversion is bandwidth-optimal in this case.
Therefore, in the rest of this section, we only consider $\Fr < \Ik$.
Moreover, in the case where $\Ir \geq \Fr$ (\regimetwo), \accessoptimal \codenames (for which explicit constructions are known) are also \bwoptimal.
Therefore, we focus on the case $\Ir < \Fr$ (\regimeone).

We start by describing the base code used in our construction, followed by the design of piggybacks, and then describe the conversion procedure along with the role of piggybacks during conversion.

\paragraph{Base code for piggybacking}
As the base code for our construction, we use a \textit{punctured} initial code of an \accessoptimal \ParamCombCode{\Ik + \Fr}{\Ik}{\Fn}{\Fk}.
Any \accessoptimal \codename can be used.
However, as mentioned in \cref{sec:background:convertible}, we assume without loss of generality that this \codename is: (1) systematic, (2) linear, and (3) only requires accessing the first $\Fr$ parities from each initial \stripe during \accessoptimal conversion.
We refer to the $[\Ik + \Fr, \Ik]$ initial code of this \accessoptimal \codename as $\Ibase$, to its $[\Fn, \Fk]$ final code as $\Fbase$. 
Let $\puncIbase$ be the punctured version of $\Ibase$ where the last $(\Fr - \Ir)$ parity symbols are punctured.

\begin{figure*}
    \colorlet{color1}{Dark2-D}
    \colorlet{color2}{Dark2-E}
    \colorlet{color3}{Dark2-C}
    \colorlet{color4}{Dark2-B}
    \renewcommand{\arraystretch}{1.3}
    \centering
    \columnwidth=\linewidth
    \newcommand{\shaderow}[1]{\rowcolor{#1!10}}
    \begin{equation*}
        \begin{array}{|c|c|}
            \multicolumn{2}{c}{\text{initial \stripe\ } 1\ (\Ivcode)} \\ \hline
            a_1 & b_1 \\ \hline
            \cvdots & \cvdots \\ \hline
            a_4 & b_4 \\ \hline
            \shaderow{color2}
            \IPP{1}{a}{1} & \IPP{1}{b}{1} \IPPA{2}{a}{1} \\ \hline
        \end{array}
        \quad
        \begin{array}{|c|c|}
            \multicolumn{2}{c}{\text{initial \stripe\ } 2\ (\Ivcode)} \\ \hline
            a_5 & b_5 \\ \hline
            \cvdots & \cvdots \\ \hline
            a_8 & b_8 \\ \hline
            \shaderow{color2}
            \IPP{1}{a}{2} & \IPP{1}{b}{2} \IPPA{2}{a}{2} \\ \hline
        \end{array}
        \quad
        \begin{array}{|c|c|}
            \multicolumn{2}{c}{\text{final \stripe\ } (\Fvcode)} \\ \hline
            a_1 & b_1 \\ \hline
            \cvdots & \cvdots \\ \hline
            a_4 & b_4 \\ \hline
            a_5 & b_5 \\ \hline
            \cvdots & \cvdots \\ \hline
            a_8 & b_8 \\ \hline
            \shaderow{color3}
            \FPP{1}{a} & \FPP{1}{b} \\ \hline
            \shaderow{color3}
            \textcolor{Dark2-B}{\FPP{2}{a}} & \FPP{2}{b} \\ \hline
        \end{array}
    \end{equation*}
    \caption{
        Example of a bandwidth-optimal \ParamCode{5}{4}{10}{8}.
        Each block in this diagram represents a \stripe, where each column corresponds to a distinct coordinate of the $\alpha$-length vector ($\alpha = 2$ in this case), and each row corresponds to a node.
        The shaded rows correspond to retired nodes for the first two blocks (initial \stripes), and new nodes for the third block (final \stripe).
        For the initial \stripes, text color is used emphasize the piggybacks.
        In the final \stripe, text color is used to denote the base code symbol which can be directly computed from the piggybacks.
    }
    \label{fig:merge:construction}
\end{figure*}

\paragraph{Piggyback design}
Now, we describe how to construct the $[\In, \Ik, \alpha]$ initial vector code $\Ivcode$ and the $[\Fn, \Fk, \alpha]$ final vector code $\Fvcode$ that make up the \bwoptimal \ParamCombCodeDefault.

The first step is to choose the value of $\alpha$.
Let us reexamine the lower bound derived in \cref{thm:merge-bw:lower-bound} for $\Ir < \Fr < \Ik$, which is rewritten below in a different form.
\[
    \gamma \geq \Is \left(\Ir \alpha +  \Ik \left( 1 - \frac{\Ir}{\Fr} \right) \alpha  \right) + \Fr\alpha.
\]
We can see that one way to achieve this lower bound would be to download exactly $\beta_1 = \alpha$ \subsyms from each of the $\Ir$ retired nodes in the $\Is$ initial \stripes, and to download $\beta_2 = \left( 1 - \sfrac{\Ir}{\Fr} \right) \alpha$ \subsyms from each of the $\Ik$ unchanged nodes in the $\Is$ initial stripes.
Thus, we choose $\alpha = \Fr$, which is the smallest value that makes $\beta_1$ and $\beta_2$ integers, thus making:
\[
    \beta_1 = \Fr \qquad\text{and}\qquad \beta_2 = (\Fr - \Ir).
\]

The next step is to design the piggybacks.
We first provide the intuition behind the design. Recall from above that we can download $\beta_2 = (\Fr - \Ir)$ \subsyms from each unchanged node and all the $\alpha$ \subsyms from each retired node.
Hence, we can utilize up to $\beta_2 = (\Fr - \Ir)$ coordinates from each of the $\Ir$ parity nodes for piggybacking.
Given that there are precisely $(\Fr - \Ir)$ punctured symbols and $\alpha$ instances of $\puncIbase$, we can store piggybacks corresponding to $\Ir$ instances of each of these punctured symbols.
During conversion, these punctured symbols can be reconstructed and used for constructing the new nodes.

Consider a message $\Msg \in \Field{q}^{\Is\Ik\alpha}$ split into $\Is\alpha$ submessages $\Mcc{s}{j} \in \Field{q}^{\Ik}$, representing the data encoded by instance $j \in [\alpha]$ of the base code in initial \stripe $s \in [\Is]$.
Recall that $\puncIbase$ is systematic by construction.
Therefore, the submessage $\Mcc{s}{j}$ will correspond to the contents of the $j$-th coordinate of the $\Ik$ systematic nodes in initial \stripe $s$.
Let $\Icc{i}{j}(s)$ denote the contents of the $j$-th coordinate of parity symbol $i$ in initial \stripe $s$ under code $\Ivcode$, and $\Fcc{i}{j}$ let denote the same for the single final \stripe encoded under $\Fvcode$.
These are constructed as follows:
\begin{align*}
    \Icc{i}{j}(s) &=
    \begin{cases}
        \Mcc{s}{j} \V{p}^I_{i}, &
        \text{for } s \in [\Is],\; i \in [\Ir],\;  1 \leq j \leq \Ir \\
        \Mcc{s}{j} \V{p}^I_{i} + \Mcc{s}{i} \V{p}^I_{j}, &
        \text{for } s \in [\Is],\; i \in [\Ir],\; \Ir < j \leq \Fr
    \end{cases}
    \\\label{eq:construction:piggy}
    \Fcc{i}{j} &= [\Mcc{1}{j} \cdots \Mcc{\Is}{j}] \V{p}^F_{i}, \qquad \text{for } i \in [\Fr],\; j \in [\Fr],
\end{align*}
where $\V{p}^I_i$ corresponds to the encoding vector of the $i$-th parity of $\Ibase$ and $\V{p}^F_i$ corresponds to the encoding vector of the $i$-th parity of $\Fbase$.
By using the \accessoptimal conversion procedure from the base code, we can compute $ \Fcc{i}{j} = [\Mcc{1}{j} \cdots \Mcc{\Is}{j}] \V{p}^F_{i}$ from $\{\Mcc{s}{j} \V{p}^I_{i} : s \in [\Is]\}$ for all $i \in [\Fr]$ and $j \in [\Fr]$.
Notice that each initial \stripe is independent and encoded in the same way (as required).

This piggybacking design, that of using parity code symbols of the base code as piggybacks, is inspired by one of the piggybacking designs proposed in \cite{rashmi2017piggybacking}, where it is used for efficiently reconstructing failed (parity) code symbols. 

\paragraph{Conversion procedure}
Conversion proceeds as follows:
\begin{enumerate}
    \item
    Download $D = \{\Mcc{s}{j} : s \in [\Is] \text{ and } \Ir < j \leq \Fr\}$, $C_1 = \{\Icc{i}{j}(s) :  s \in [\Is],\ i \in [\Ir], \text{ and } 1 \leq j \leq \Ir\}$, and $C_2 = \{\Icc{i}{j}(s) : s \in [\Is],\ i \in [\Ir], \text{ and } \Ir < j \leq \Fr\}$.
    
    \item
    Recover the piggybacks $C_3 = \{\Mcc{s}{j} \V{p}^I_i : s \in [\Is],\ \Ir < i \leq \Fr, \text{ and } 1 \leq j \leq \Ir\}$ by computing $\Mcc{s}{i} \V{p}^I_j$ from $D$ and obtaining $\Mcc{s}{j} \V{p}^I_i = \Icc{j}{i}(s) - \Mcc{s}{i} \V{p}^I_j$ using $C_2$.
    
    \item
    Compute the remaining base code symbols from the punctured symbols $C_4 = \{\Mcc{s}{i} \V{p}^I_j : s \in [\Is],\ \Ir < i \leq \Fr, \text{ and } \Ir < j \leq \Fr\}$ using $D$.
    
    \item
    Compute the parity nodes of the final \stripe specified by the subsymbols $C_5 = \{\Fcc{i}{j} : i \in [\Fr],\ j \in [\Fr]\}$.
    This is done by using the conversion procedure from the \accessoptimal \codename used as base code to compute $C_5$ from $C_1,\ C_2,\ C_3,$ and $C_4$.
\end{enumerate}

This procedure requires downloading $\beta_1$ \subsyms from each retired node and $\beta_2$ \subsyms from each unchanged node.
Additionally, $\Fr\alpha$ network bandwidth is required to write the new nodes.
Thus, the total network bandwidth of conversion is:
\begin{align*}
    \gamma &=
    \Is\left( \Ir\beta_1 + \Ik\beta_2 \right) + \Fr\alpha \\
    &=
    \Is\left( \Ir\alpha + \Ik\left( 1 - \frac{\Ir}{\Fr} \right) \right) + \Fr\alpha
\end{align*}
which matches \cref{thm:merge-bw:lower-bound}.

Now we show a concrete example of our construction.
\begin{example}[Bandwidth-optimal conversion in the merge regime] 
    \label{ex:piggyback-conversion}
    Suppose we want to construct a \bandwidthoptimal\ \ParamCombCode{5}{4}{10}{8} over a finite field $\Field{q}$ (assume that $q$ is sufficiently large).
    As a base code, we use a punctured \accessoptimal\ \ParamCombCode{6}{4}{10}{8}.
    Thus, $\Ibase$ is a $[6,4]$ code, $\Fbase$ is a $[10, 8]$ code, and $\puncIbase$ is a $[5,4]$ code, all derived from the chosen \accessoptimal \codename as described in the construction above.
    Let \(\V{p}_1^I, \V{p}_2^I \in \Field{q}^{4 \times 1}\) be the encoding vectors for the parities of $\Ibase$, and \(\V{p}_1^F, \V{p}_2^F \in \Field{q}^{8 \times 1}\) be the encoding vector for the parities of $\Fbase$.
    
    Since $\alpha = 2$, we construct a $[5, 4, 2]$ initial vector code $\Ivcode$ and a $[10, 8, 2]$ final vector code $\Fvcode$.
    Let $\V{a} = (a_1, \ldots, a_8)$ and $\V{b} = (b_1, \ldots, b_8)$.
    \Cref{fig:merge:construction} shows the resulting piggybacked codes encoding submessages $\V{a}^{(1)} = (a_1, \ldots, a_4), \V{a}^{(2)} = (a_5, \ldots, a_8), \V{b}^{(1)} = (b_1, \ldots, b_4), \V{b}^{(2)} = (b_5, \ldots, b_8) \in \Field{q}^{1 \times 4}$.
    
    During conversion, only 12 \subsyms need to be downloaded: \(\V{b}^{(1)}, \V{b}^{(2)}\) and all the parity symbols from both \stripes.
    From these subsymbols, we can recover the piggyback terms \(\V{a}^{(1)} \V{p}_2^I\) and \(\V{a}^{(2)} \V{p}_2^I\), and then compute $\V{b}^{(1)} \V{p}_2^I$ and $\V{b}^{(2)} \V{p}_2^I$ in order to reconstruct the second parity symbol of $\Ibase$.
    Finally, we use $\V{a}^{(i)} \V{p}_1^I, \V{b}^{(i)} \V{p}_1^I, \V{a}^{(i)} \V{p}_2^I, \V{b}^{(i)} \V{p}_2^I$ for $i \in \{1,2\}$ with the conversion procedure from the access-optimal \codename to compute the base code symbols $\V{a}\ \V{p}^F_1, \V{a}\ \V{p}^F_2, \V{b}\ \V{p}^F_1$ and $\V{b}\ \V{p}^F_2$ of the new nodes.
    
    The default approach would require one to download 16 \subsyms in total from the initial nodes.
    Both approaches require downloading 4 \subsyms in total from the coordinator node to the new nodes.
    Thus, the proposed construction leads to $20\%$ reduction in \convbw as compared to the default approach of reencoding.
\end{example}

\subsection{Convertible codes with bandwidth-optimal conversion for multiple final parameters}
\label{sec:merge-bw:construction:multi}

In practice, the final parameters $\Fn, \Fk$ might depend on observations made after the initial encoding of the data and hence they may be unknown at code construction time.
In particular, for a \ParamCombCodeDefault in the merge regime this means that the values of $\Is = \Cs$ and $\Fr = (\Fn - \Fk)$ are unknown.

To ameliorate this problem, we now present convertible codes which support \bwoptimal conversion \emph{simultaneously} for multiple possible values of the final parameters.
Recall property (4) of the access-optimal base code which we reviewed in \cref{sec:background:convertible}: when constructed with a given value of $\Is = \Cs$ and $\Fr = r$, the initial $[\In, \Ik]$ code is \accessConvertibleDefault for all $\Fk = \Cs'\Ik$ and $\Fn = \Fk + r'$ such that $1 \leq \Cs' \leq \Cs$ and $1 \leq r' \leq r$.

\subsubsection{Supporting multiple values of \texorpdfstring{$\Is$}{initial lambda}}
The construction from \cref{sec:merge-bw:construction} for some particular value of $\Is = \Cs$, natively supports bandwidth-optimal conversion for any $\Is = \Cs' < \Cs$.
This is a consequence of property (4) above, and can be done easily by considering one or multiple of the initial \stripes as consisting of zeroes only, and ignoring them during conversion.
From \cref{thm:merge-bw:lower-bound}, it is easy to see that this modified conversion procedure achieves optimal network bandwidth cost for the new parameter $\Is = \Cs'$.

\subsubsection{Supporting multiple values of \texorpdfstring{$\Fr$}{final r}}
We break this scenario into two cases:

\textbf{Case 1} (supporting $\Fr \leq \Ir$): due to property (4) above, the base code used in the construction from \cref{sec:merge-bw:construction} supports access-optimal conversion for any value of $\Fr = r$ such that $r \leq \Ir$.
Using this property, one can achieve bandwidth optimality for any $r \leq \Ir$ by simply using the access-optimal conversion on each of the $\alpha$ instances of the base code independently.
The only difference is that some of the instances might have piggybacks, which can be simply ignored.
The final code might still have these piggybacks, however they will still satisfy the property that the piggybacks in instance $i$ ($2 \leq i \leq \alpha$) only depend on data from instances $\{1, \ldots, (i - 1)\}$.
Thus, the final code will have the MDS property and the desired parameters.

\textbf{Case 2} (supporting $\Fr > \Ir$): for supporting multiple values of $\Fr \in \{r_1, r_2, \ldots, r_s\}$ such that $r_i > \Ir$ ($i \in [s]$), we start with an \accessoptimal \codename having $\Fr = \max_i r_i$.
Then we repeat the piggybacking step of the construction (see \cref{sec:merge-bw:construction:single}) for each $r_i$, using the resulting code from step $i$ {(with the punctured symbols from $\Ibase$ added back)} as a base code for step $(i + 1)$.
Therefore, the resulting code will have $\alpha = \prod_{i=1}^s r_i$.
Since the piggybacking step will preserve the MDS property of its base code, and the initial code used in the first piggyback step is MDS, it is clear that the initial code resulting from the last piggybacking step will also be MDS.
Conversion for one of the supported $\Fr = r_i$ is performed as described in \cref{sec:merge-bw:construction:single} on each of the additional instances created by steps $(i + 1), \ldots, s$ (i.e.\ $\prod_{i' = (i + 1)}^s r_{i'}$ in total).
As before, some of these instances {after conversion} will have piggybacks, which can be simply ignored, as the resulting code will continue to have the property that piggybacks from a given instance only depend on data from earlier instances.

\begin{figure*}
    \centering
    \colorlet{color1}{Dark2-D}
    \colorlet{color2}{Dark2-E}
    \colorlet{color3}{Dark2-C}
    \colorlet{color4}{Dark2-B}
    \colorlet{color5}{Dark2-D}
    \arraycolsep=2pt 
    \renewcommand{\arraystretch}{1.3}
    \centering
    \columnwidth=\linewidth
    \newcommand{\shaderow}[1]{\rowcolor{#1!10}}
    \newcommand{\PP}[2]{\V{p}_{#1}^T \V{#2}^{\scaleto{(1)}{5pt}}}
    \newcommand{\PPA}[2]{\textcolor{color4}{{} + \V{p}_{#1}^T \V{#2}^{\scaleto{(1)}{5pt}}}}
    \newcommand{\PPB}[2]{\textcolor{color5}{{} + \V{p}_{#1}^T \V{#2}^{\scaleto{(1)}{5pt}}}}
    \footnotesize
    \begin{equation*}
        \begin{array}{|c|c||c|c||c|c|}
            \multicolumn{6}{c}{\text{initial \stripe\ } 1\ (\Ivcode)} \\ \hline
            a_1 & b_1 &
            c_1 & d_1 &
            e_1 & f_1 \\ \hline
            \cvdots & \cvdots &
            \cvdots & \cvdots &
            \cvdots & \cvdots \\ \hline
            a_4 & b_4 &
            c_4 & d_4 &
            e_4 & f_4 \\ \hline
            \shaderow{color2}
            \IPP{1}{a}{1} &
            \IPP{1}{b}{1} \IPPA{2}{a}{1} &
            \IPP{1}{c}{1} \IPPB{2}{a}{1} &
            \IPP{1}{d}{1} \IPPA{2}{c}{1} \IPPB{2}{b}{1} &
            \IPP{1}{e}{1} \IPPB{3}{a}{1} &
            \IPP{1}{f}{1} \IPPA{2}{e}{1} \IPPB{3}{b}{1} \\ \hline
        \end{array}
    \end{equation*}
    \begin{equation*}
        \begin{array}{|c|c||c|c||c|c|}
            \multicolumn{6}{c}{\text{final \stripe\ } (\Fr = 2)}\\ \hline
            a_1 & b_1 &
            c_1 & d_1 &
            e_1 & f_1 \\ \hline
            \cvdots & \cvdots &
            \cvdots & \cvdots &
            \cvdots & \cvdots \\ \hline
            a_8 & b_8 &
            c_8 & d_8 &
            e_8 & f_8 \\ \hline
            \shaderow{color3}
            \FPP{1}{a} &
            \FPP{1}{b} &
            \FPP{1}{c} \IPPB{2}{a}{1} \IPPB{2}{a}{2} &
            \FPP{1}{d} &
            \FPP{1}{e} \IPPB{3}{a}{1} \IPPB{3}{a}{2} &
            \FPP{1}{f} \\ \hline
            \shaderow{color3}
            \FPPA{2}{a} &
            \FPP{2}{b} &
            \FPPA{2}{c} &
            \FPP{2}{d}&
            \FPPA{2}{e} &
            \FPP{2}{f} \\ \hline
        \end{array}
    \end{equation*}
    \begin{equation*}
        \begin{array}{|c|c||c|c||c|c|}
            \multicolumn{6}{c}{\text{final \stripe\ } (\Fr = 3)}\\ \hline
            a_1 & b_1 &
            c_1 & d_1 &
            e_1 & f_1 \\ \hline
            \cvdots & \cvdots &
            \cvdots & \cvdots &
            \cvdots & \cvdots \\ \hline
            a_8 & b_8 &
            c_8 & d_8 &
            e_8 & f_8 \\ \hline
            \shaderow{color3}
            \FPP{1}{a} &
            \FPP{1}{b} \IPPA{2}{a}{1} \IPPA{2}{a}{2} &
            \FPP{1}{c} &
            \FPP{1}{d} &
            \FPP{1}{e} &
            \FPP{1}{f} \\ \hline
            \shaderow{color3}
            \FPPB{2}{a} &
            \FPPB{2}{b} &
            \FPP{2}{c} &
            \FPP{2}{d} &
            \FPP{2}{e} &
            \FPP{2}{f} \\ \hline
            \shaderow{color3}
            \FPPB{3}{a} &
            \FPPB{3}{b} &
            \FPP{3}{c} &
            \FPP{3}{d} &
            \FPP{3}{e} &
            \FPP{3}{f} \\ \hline
        \end{array}
    \end{equation*}
    \caption{
    Example of a $[5,4]$ MDS code that supports \bwoptimal conversion to multiple final codes (only one initial stripe is shown).
    This code supports \bandwidthoptimal conversion to a $[8 + r, 8]$ MDS code for $r = 1, 2, 3$.
    Text color is used in the initial \stripe to denote piggybacks from different piggybacking steps.
    In the possible final \stripes, text color is used to denote base code symbols which are directly computed from the corresponding piggybacks, or to denote leftover piggybacks which were not used during conversion.
    }
    \label{fig:construction:multi}
\end{figure*}

\begin{example}[\bwoptimal conversion for multiple final parameters]
    In this example, we will extend the \ParamCode{5}{4}{10}{8} from \cref{ex:piggyback-conversion} ($\Fr = 2$) to construct a code which additionally supports \bwoptimal conversion to an $[11, 8]$ MDS code ($\Fr = 3$).
    \Cref{fig:construction:multi} shows one initial \stripe of the new initial vector code, which has $\alpha = 2 \cdot 3 = 6$.
    Here $\V{a}^{(1)} = (a_1, \ldots, a_4)$, $\V{a}^{(2)} = (a_5, \ldots, a_8)\in \Field{q}^{1 \times 4}$, $\V{a} = (a_1, \ldots, a_8) \in \Field{q}^{1 \times 8}$, and similarly for $\V{b}, \ldots, \V{f}$.
    The vectors $\V{p}^I_i \in \Field{q}^{4 \times 1}$ are the encoding vectors of the initial code $\Ibase$ and $\V{p}^F_i \in \Field{q}^{8 \times 1}$ are encoding vectors of the final code $\Fbase$ ($i \in \{1, 2, 3\}$). 
    Since the maximum supported $\Fr$ is 3, we start with an \accessoptimal \ParamCode{7}{4}{11}{8}.
    Thus, $\Ibase$ is a $[7, 4]$ code, $\Fbase$ is a $[11, 8]$ code, and $\puncIbase$ is a $[5, 4]$ code.
    In the first round of piggybacking we consider $\Fr = 2$, which yields the code shown in \cref{ex:piggyback-conversion}.
    In the second round of piggybacking we consider $\Fr = 3$ and piggyback the code resulting from the first round, which yields the code shown in \cref{fig:construction:multi}.
    Conversion for $\Fr = 1$ proceeds by simply downloading the contents of the single parity node and using the \accessoptimal conversion procedure.
    Conversion for $\Fr = 2$ proceeds by treating this code as three instances of the code from \cref{ex:piggyback-conversion} and performing conversion for each one independently.
    Conversion for $\Fr = 3$ proceeds by treating this code as a vector code with $\alpha = 3$ and base field $\Field{q^2}$ (i.e.\ each element is a vector over $\Field{q}$ of length 2).
\end{example}

\begin{remark}[Field size requirement]
    The field size requirement for $\Field{q}$ of the constructions presented in this section is given by the field size requirement of the base code used.
    The currently lowest known field size requirement for an explicit construction of systematic linear \accessoptimal \codenames in the merge regime is given by \cite{maturana2020convertible}.
    In general, this requirement is roughly $q \geq 2^{\varsigma (\In)^3}$.
    When $\Fr \leq \Ir - \Cs + 1$, the requirement can be significantly reduced to $q \geq \Ik\Ir$.
    And when $\Fr \leq \lfloor \sfrac{\Ir}{\Cs} \rfloor$, the requirement can be further reduced to $q \geq \max\{\In, \Fn\}$.
\end{remark}

%% file: sections/7-savings.tex
\section{Bandwidth savings of bandwidth-optimal \codenames}
\label{sec:evaluation}

\begin{figure}
    \centering
    \includegraphics[width=.5\textwidth]{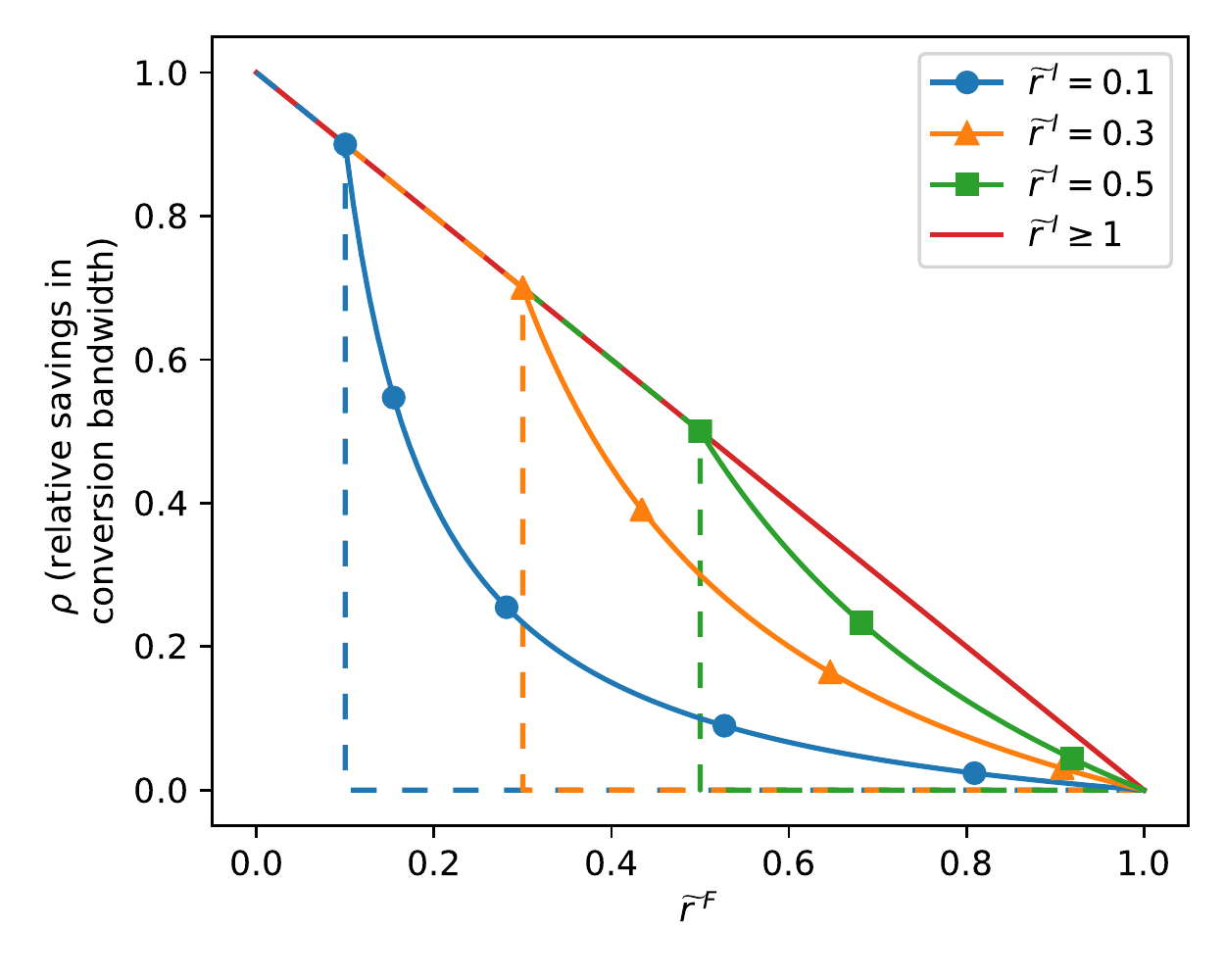}
    \caption{
        Achievable savings in conversion bandwidth by \bwoptimal \codenames in comparison to the default approach to conversion.
        Here \(\Irt = \Ir / \Ik\) and \(\Frt = \Fr / \Ik\) are the initial and final redundancies, divided by the initial code dimension.
        Each curve shows the relative savings for a fixed value of \(\Irt\), as \(\Frt\) varies.
        Solid lines indicate \bwoptimal \codenames, and dashed lines indicate \accessoptimal \codenames.
        Notice that each curve overlaps with the red curve ($\Irt \geq 1$) in the range $\Frt \in (0, \Irt]$.
    }
    \label{fig:savings}
\end{figure}

In this section, we show the amount of savings in bandwidth that can be obtained by using \bwoptimal \codenames in the merge regime, relative to the default approach to conversion.
We present the amount of savings in terms of two ratios:
\[
    \Irt = (\sfrac{\Ir}{\Ik}) \qquad\text{and}\qquad \Frt = (\sfrac{\Fr}{\Ik}),
\]
i.e.\ the initial and final amount of ``redundancy'' relative to the initial dimension of the code.
For simplicity, we only consider the \bwcost of communication from nodes to the coordinator node, since the \bw cost of communication from the coordinator node to new nodes is fixed for stable \codenames (specifically, it is equal to $\alpha\Fr$).
Thus, the \bwcost of the default approach is always $\Cs\Ik\alpha$.
\Cref{fig:savings} shows the relative savings, i.e.\ the ratio between the \bwcost of optimal conversion and the \bwcost of conversion under the default approach, for fixed values of \(\Irt \in (0, \infty)\) and varying \(\Frt \in (0, \infty)\).

Each curve shown in \cref{fig:savings} can be divided into three regions, depending on the value of $\Frt$:
\begin{itemize}
    \item \textbf{Region \(0 < \Frt \leq \Irt\) and $\Frt < 1$:} this implies that \(\Fr \leq \Ir\), so by \cref{thm:merge-bw:lower-bound:one} the conversion bandwidth is \(\Cs\Fr\alpha\), and the relative savings are:
    \[
        \rho = 1 - \frac{\Cs\Fr\alpha}{\Cs\Ik\alpha} = 1 - \Frt.
    \]
    This region corresponds to \regimetwo, and in this region \accessoptimal \codenames are also \bwoptimal.
    This region of the curve is linear, and the amount of savings is not affected by $\Irt$.
    
    \item \textbf{Region \(\Irt < \Frt < 1\):} this implies that \(\Ir \leq \Fr \leq \Ik\), and by \cref{thm:merge-bw:lower-bound:two} the conversion bandwidth is \(\Cs\alpha(\Ir + \Ik(1 - \Ir / \Fr))\), and the relative savings are:
    \[
        \rho = 
        1 - \frac{\Cs\alpha \left( \Ir + \Ik\left( 1 - \frac{\Ir}{\Fr} \right) \right)}{\Cs\Ik\alpha}
        =
        \Irt \left( \frac{1}{\Frt} - 1 \right).
    \]
    This corresponds to \regimeone, where access-optimal convertible codes provide no \convbw savings.
    Thus \bwoptimal \codenames provide substantial savings in \convbw in this regime, compared to \accessoptimal \codenames.
    
    \item \textbf{Region \(\Frt \geq 1\):} this implies that \(\Fr \geq \Ik\) and by \cref{thm:merge-bw:lower-bound:one} a bandwidth of \(\Cs\Ik\alpha\) is required.
    Thus no savings in bandwidth cost are possible in this region.
\end{itemize}

Thus, bandwidth-optimal convertible codes allow for savings in network bandwidth on a much broader region relative to access-optimal convertible codes.

%% file: sections/8-conclusion.tex
\section{Conclusions and future directions}
\label{sec:conclusion}
In this paper, we initiated a study on the network bandwidth cost of convertible codes.
We showed that the conversion problem can be effectively modeled using network information flow to obtain lower bounds on \convbw.
Using the bounds derived, we showed that for the merge regime \accessoptimal \codenames are also \bwoptimal when $\Ir \geq \Fr$ (\regimeone) and that there is room for reducing \convbw when $\Ir < \Fr$ (\regimetwo).
We proposed an explicit construction which achieves the optimal \convbw for all parameters in the merge regime.
Finally, we showed that \bwoptimal \codenames can achieve substantial savings in \convbw over the default approach and \accessoptimal \codenames.

This work leads to several open questions and challenges.
The main challenge is to extend the \convbw lower bounds and \bwoptimal constructions to encompass all possible parameter values (i.e. the general regime).
Another important challenge is characterizing the optimal value of $\alpha$, especially in the case of multiple possible final parameter values, where $\alpha$ can become very large when using the construction proposed in this paper.
Yet another open challenge is lowering the field size requirement of bandwidth-optimal \codename constructions, as well as deriving lower bounds for their field size requirements.